\documentclass[12pt,reqno]{amsart}
\usepackage{graphicx}
\usepackage{amssymb,amsmath}
\usepackage{amsthm}
\usepackage{color}
\usepackage[pdf]{pstricks}
\usepackage{hyperref}

\usepackage{amssymb}
\usepackage{epsfig}
\usepackage{extarrows}
\usepackage{amsfonts}
\usepackage{graphicx}

\newcommand{\R}{\mathbb{R}}

\makeatletter\@addtoreset {equation}{section}\makeatother

\newtheorem{remark}{Remark}
\newtheorem{proposition}{Proposition}
\newtheorem{definition}{Definition}
\newtheorem{theorem}{Theorem}
\newtheorem{corollary}{Corollary}

\begin{document}
	
\title[Breathers in the nonlocal derivative NLS equation]{\bf Traveling periodic waves and breathers \\ in the nonlocal derivative NLS equation}
	
\author{Jinbing Chen}
\address[J. Chen]{School of Mathematics, Southeast University, Nanjing, Jiangsu 210096, PR China}
\email{cjb@seu.edu.cn}

\author{Dmitry E. Pelinovsky}
\address[D. Pelinovsky]{Department of Mathematics and Statistics, McMaster University, Hamilton, Ontario, Canada, L8S 4K1}
\email{pelinod@mcmaster.ca}

\begin{abstract}
A nonlocal derivative NLS (nonlinear Schr\"{o}dinger) equation describes 
modulations of waves in a stratified fluid and a continuous limit of the Calogero--Moser--Sutherland system of particles. For the defocusing version of this equation, we prove the linear stability of the nonzero constant background for decaying and periodic perturbations and the nonlinear stability for periodic perturbations. For the focusing version of this equation, we prove linear and nonlinear stability of the nonzero constant background under some restrictions. For both versions, we characterize the traveling periodic wave solutions by using Hirota's bilinear method, both on the nonzero and zero backgrounds. For each family of traveling periodic waves, we construct families of breathers which describe solitary waves moving across the stable background. A general breather solution with $N$ solitary waves propagating on the traveling periodic wave background is derived in a closed determinant form.
\end{abstract}

\date{\today}
\maketitle

\section{Introduction}
\label{sec-1}

One of the main models of the modern nonlinear physics is the nonlinear Schr\"{o}dinger (NLS) equation which describes a slow modulation of small-amplitude, nearly harmonic waves \cite{Fibich,Kev-Dark-2015}. Due to the scaling transformation, the cubic NLS equation can be written in the 
dimensionless form 
\begin{equation}
i u_t = -u_{xx} \pm |u|^2 u, \qquad u(x,t) : \mathbb{R}\times \mathbb{R} \mapsto \mathbb{C},
\label{NLS}
\tag{NLS$_{\pm}$}
\end{equation}
where the upper and lower signs correspond to the defocusing and focusing versions of this model, respectively. For the defocusing NLS$_+$ equation, localized perturbations of the zero background scatter to zero as the time evolves and stable dark solitons propagate on a nonzero constant background which is linearly and nonlinearly stable in the time evolution. For the focusing NLS$_-$ equation, stable bright solitons propagate on the zero background  due to the balance between nonlinearity and dispersion whereas the nonzero constant background is unstable with rogue waves appearing from nowhere and disappearing without any trace. These phenomena were reviewed in \cite{Suret,Amin,Review}. 

The purpose of this work is to consider stability of the nonzero constant background and propagation of solitary waves on the traveling periodic wave background in the nonlocal derivative NLS equation. Due to the scaling transformation, the model  can be written in the dimensionless form
\begin{equation}
	\label{2.1}
	i u_t =  u_{xx} \pm u (i + H) (|u|^2)_x,  \qquad u(x,t) : \mathbb{R}\times \mathbb{R} \mapsto \mathbb{C},
	\tag{NDNLS$_{\pm}$}
\end{equation}
where $H$ is the Hilbert transform which can be defined on $\R$ either according to the following integral formula, 
$$
H(f) := \frac{1}{\pi} {\rm p. v.} \int_{-\infty}^{\infty} \frac{f(y) dy}{y-x},
$$
or according to the Fourier transform by $H(e^{ikx}) = i {\rm sgn}(k) e^{ikx}$, $k \in \R$.

The NDNLS$_+$ equation was derived in \cite{P95,PG96} in the context of modulation theory for internal waves in a stratified fluid. It appears as an asymptotic reduction of the integrable Benjamin--Ono (BO) equation for the envelope of modulating wave packets. A more general intermediate NLS equation was derived from an intermediate long-wave equation in \cite{P95}. The intermediate NLS equation connects NLS$_+$ in the limit of shallow fluid and NDNLS$_+$ in the limit of deep fluid. Integrability and existence of the Lax pair for the NDNLS$_+$ equation was established in \cite{PG95} and was used in many studies of periodic and solitary waves  \cite{MatsunoPLA,MatsunoKaup,Matsuno,MatsunoJCPJ}.

The NDNLS$_-$ equation was obtained as a continuum limit for dynamics of particles in the Calogero--Moser--Sutherland (CMS) system \cite[Eq. (40)]{Abanov1}, see also \cite{Voit} for review of the CMS system. The NDNLS$_-$ equation is related to a Hamiltonian formulation of the complex extension of the BO  equation with a bi-directional wave propagation. The recent interest to this model was inspired by applications of methods of harmonic analysis to weak turbulence in the integrable nonlocal equations with the localized data on the infinite line \cite{GL-22}. Global well-posedness and blow-up in infinite or finite time were studied in \cite{HK,KLV,KKK,KK,UK}. Traveling periodic waves and dynamics of the initial-value problem in the periodic domain were analyzed in \cite{B-22,B-23,B-24}. Numerical approximations of the model with the spectral method were developed in \cite{Bronsard}. Coupled systems of nonlocal integrable equations were also discussed in \cite{BLL-22,BF-23,Sun-24}.

In analogy with the cubic NLS$_{\pm}$ equations, the two versions of the NDNLS$_{\pm}$ equations are referred to as ``defocusing" (upper sign) and ``focusing" (lower sign). However, we will show here that the nonzero constant background does not exhibit the rogue wave phenomenon. Solitary waves propagate steadily on the nonzero constant background as dark solitons for the ``defocusing" version \cite{MatsunoPLA,P95} and bright solitons for the ``focusing" version \cite{GL-22,MatsunoSAMP}. The linear stability of the nonzero constant background for decaying and periodic perturbations and the nonlinear stability for periodic perturbations is proven for the NDNLS$_+$ equation. The proof of the linear and nonlinear stability for the NDNLS$_-$ equation holds only under some restrictions. The main results of this work are given by Theorems \ref{prop-stab}, \ref{theorem-stab}, \ref{prop-3}, \ref{prop-tech}, and Corollaries \ref{cor-stab}, \ref{prop-foc-5}.

The travelling periodic wave background has been recently studied in the cubic NLS equation \cite{Biondini,Biondini2,El} and the  Korteweg--de Vries (KdV) equation \cite{Hoefer,Girotti,Grava} because it arises naturally due to the gradient catastrophe of the wave profiles in the limit of small 
dispersion. For the defocusing cubic NLS equation, dark solitons propagate on the stable traveling periodic wave background 
\cite{Ling,Shin,Takahashi} (see also \cite{Bertola,HMP} for the KdV equation and \cite{MP-24} for the defocusing modified KdV equation). For the focusing cubic NLS equation, bright breathers and rogue waves arise on the unstable traveling periodic wave background \cite{CPW,Feng} (see also \cite{CP-dnls,CPU-dnls} for the cubic derivative NLS equation and 
\cite{CP-23,CP-24} for the discrete NLS and discrete modified KdV equations).

We will show that solitary waves propagate steadily on the traveling periodic wave background in both versions of the NDNLS$_{\pm}$ equations. This suggests that the physics terminology of the two versions as ``defocusing" and ``focusing" is not justified. We conjecture that the traveling periodic wave is linearly stable with respect to small perturbations. The latter question is left open for further studies. Compared to the propagation of solitary waves on the elliptic traveling wave background in the NLS, KdV, and modified KdV equations, the corresponding solutions for the NDNLS$_{\pm}$ equations are expressed by the elementary (trigonometric and power) functions. 

Propagation of solitary waves in the NDNLS$_{\pm}$ equations is very similar to the one in the BO equation 
explored in our previous work \cite{ChenPel24}. For the defocusing NDNLS$_{+}$ equation, we only obtain dark solitons propagating 
on the traveling periodic wave background. For the focusing NDNLS$_{-}$ equation, we obtain both bright and dark solitons on the traveling periodic wave background. The characteristic properties of such solutions depend 
on the Lax spectrum associated with the traveling periodic wave explored recently in 
\cite{B-23} based on earlier work \cite{GK-21,GKT-20,GKT-22} for the BO equation. From the technical point of view, we rely on the Hirota bilinear form both for the NDNLS$_{\pm}$ equations and their Lax pair representation \cite{MatsunoPLA,MatsunoJCPJ,MatsunoSAMP}. By degeneration of the multi-periodic 
solutions, we obtain a closed determinant form for $N$ solitary waves propagating on the traveling periodic wave background.

The paper is organized as follows. Section \ref{sec-2} contains preliminary facts about the NDNLS$_{\pm}$ equations. Linear and nonlinear stability of the nonzero constant solution is considered in Section \ref{sec-3}. The traveling periodic waves, their Lax spectrum, and breathers on their background are  obtained in Sections \ref{sec-4} and \ref{sec-5} for nonzero and zero constant backgrounds, respectively. Section \ref{sec-concl} concludes the paper.

\section{Lax pair for the nonlocal derivative NLS equations}
\label{sec-2}

Both versions of the nonlocal derivative NLS equations can be rewritten as 
\begin{equation}
\label{INLS}
i u_t =  u_{xx} + \sigma u (i + H) (|u|^2)_x, \qquad u(x,t) : \mathbb{R}\times \mathbb{R} \mapsto \mathbb{C},
\end{equation}
where the sign parameter $\sigma$ is
$$
\sigma = +1 \;\;\Leftrightarrow \;\; \mbox{\rm ``defocusing"} \quad \mbox{\rm and} \quad 
\sigma = -1 \;\;\Leftrightarrow \;\;  \mbox{\rm ``focusing"}.
$$
Several symmetries are identical between the NLS$_{\pm}$ and NDNLS$_{\pm}$ equations. The list includes the translational and rotational symmetries 
\begin{equation}
\label{NLS-symm}
u(x,t)  \mapsto u(x+x_0,t+t_0) e^{i \theta_0}, 
\qquad x_0, t_0, \theta_0 \in \mathbb{R},
\end{equation}
the Lorentz transformation
\begin{equation}
\label{NLS-Lorentz}
u(x,t) \mapsto 
e^{-\frac{i}{2} cx + \frac{i}{4} c^2 t} u(x-ct,t), 
\qquad c \in \mathbb{R},
\end{equation}
and the scaling symmetry
\begin{equation}
\label{NLS-scaling}
u(x,t) \mapsto \alpha u(\alpha^2 x, \alpha^4 t), \quad \alpha > 0.
\end{equation}
The symmetry transformation (\ref{NLS-symm}) can be used to set the translational parameters of the traveling periodic wave to zero. The Lorentz transformation (\ref{NLS-Lorentz}) can be used to normalize the profile 
of the traveling waves as ${\rm Im}(x) \to \pm \infty$, see expressions (\ref{2.10}) and (\ref{3.14}) in Propositions \ref{prop-1} and \ref{prop-foc-1}. The scaling transformation (\ref{NLS-scaling}) can be used to normalize the nonzero constant background to unity, see expression (\ref{2.3}), or the period of the traveling periodic wave on the zero background to $2\pi$, see Remark \ref{rem-unity}. 

The nonlocal model (\ref{INLS}) 
is a compatibility condition of the following linear system
\begin{equation}
	\label{2.2}
	\left\{ 
	\begin{array}{l}
		i p_x + \lambda p + u q^+ = 0,\\
		q^+ - \mu q^- + \sigma \bar{u} p =0,\\
		i p_t + \lambda^2 p + \lambda u q^+ + i (u q_x^+ - u_x q^+) = 0,\\
		i q_t^\pm - 2 i \lambda q_x^\pm + q_{xx}^\pm + \sigma q^{\pm} [(\pm i + H) (|u|^2)_x]  = 0,
	\end{array} \right.
\end{equation}
where $\lambda$ is the spectral parameter, $\bar{u}$ denotes the complex conjugation of $u$, and $(p,q^+,q^-)$ are components of the eigenfuction, in which $ q^+ $ and $ q^- $ are analytic in the upper and lower half of the complex plane of $x$, respectively. If we use the projection operators $\mathcal{P}^{\pm} :=  \frac12 (1 \mp i H)$, then we have 
\begin{equation}
\label{projection-formulas}
q^{\pm} =  \mathcal{P}^{\pm} q^{\pm} \quad \mbox{\rm and} \quad 0 = \mathcal{P}^{\mp} q^{\pm}.
\end{equation} 

The nonlocal model (\ref{INLS}) is obtained from the commutability condition $p_{xt} = p_{tx}$ by using the first, third, and fourth equations of system (\ref{2.2}) for $p_x$, $p_t$, and $q^+_t$. From the second equation of system (\ref{2.2}), we obtain 
\begin{equation}
\label{q-plus-minus-projection}
q^+ = - \sigma \mathcal{P}^+ (\bar{u} p) \quad \mbox{\rm and} \quad q^- = \sigma \mu^{-1} \mathcal{P}^-(\bar{u}p),
\end{equation}
where $\mu$ is an additional parameter to be determined.

Since $-i \partial_x$ and $\mathcal{P}^+$ are self-adjoint operators in $L^2$-based Hilbert spaces such as $L^2(\mathbb{R})$ or $L^2_{\rm per}$, the Lax operator $\mathcal{L}_u : H^1 \subset L^2 \to L^2$ given by 
\begin{equation}
\label{linear-operator}	
	\mathcal{L}_u := -i \partial_x + \sigma u \mathcal{P}^+(\bar{u} \; \cdot)
\end{equation}
is self-adjoint. By the spectral theorem in Hilbert space $L^2$, the Lax spectrum (the set of admissible values of the spectral parameter $\lambda$) is a subset of the real line. 

If solutions of the nonlocal model (\ref{INLS}) in $H^1$ are restricted in the space of analytic functions in $\mathbb{C}_+$, then $u \in H^1 \cap L^2_+$, where $L^2_+$ is defined by $L^2_+ := \{ u \in L^2 : \; \mathcal{P}^+ u = u \}$.
Then $u q^+$ and $u q_x^+ - u_x q^+$ are analytic functions in $\mathbb{C}_+$ so that the linear system (\ref{2.2}) can be closed 
for $p$ being also analytic in $\mathbb{C}_+$. This brings the spectrum of a restricted self-adjoint operator $\mathcal{L}_u |_{L^2_+} : H^1 \cap L^2_+ \subset L^2 \to L^2$ given by 
\begin{equation}
\label{linear-operator-restr}
	\mathcal{L}_u |_{L^2_+} = -i \partial_x + \sigma \mathcal{P}^+ u \mathcal{P}^+(\bar{u} \; \cdot)
\end{equation}
and considered in \cite{B-23,GL-22} for periodic and soliton solutions. 
In consistency with this restriction,  we will show in our work that the profile $u$ of the traveling periodic waves satisfies $u \in H^1 \cap L^2_+$ and that the component $p$ of the eigenfunction satisfies $p \in H^1 \cap L^2_+$, from which the spectra of $\mathcal{L}_u$ and $\mathcal{L}_u |_{L^2_+}$ are equivalent, see Propositions \ref{prop-1}, \ref{prop-2}, \ref{prop-foc-1}, and \ref{prop-foc-3}.

As in the case of the BO equation \cite{ChenPel24}, there are two exact solutions of the linear system (\ref{2.2}) if $u$ is the spatial profile of the traveling periodic wave. The first solution has $q^- \equiv 0$ and the second solution has both $q^+$ and $q^-$ nonzero. 

For the first solution, the admissible values of $\lambda$ in the Lax spectrum are defined by the existence of bounded components $(p,q^+,q^-=0)$ of the eigenfunction such that $q^+$ is analytic and bounded in $\mathbb{C}_+$. For the second solution, it is not sufficient to look for the bounded components $(p,q^+,q^-)$ of the eigenfunction for which $q^{\pm}$ are analytic and bounded in $\mathbb{C}_{\pm}$ \cite{ChenPel24}. Based on Propositions 2.2 and C.2 of \cite{GK-21} for the BO equation, we need to define the additional spectral bands in $\cup_{j=0}^{N} [\lambda_j,\lambda_j + k_1]$, where $\{ \lambda_0, \lambda_1, \dots,\lambda_N \}$ are admissible values of $\lambda$ for which the mean value of nonzero $q^-$ over the spatial period of $u$ is zero. The values of $\{ \lambda_0, \lambda_1, \dots,\lambda_N \}$  correspond to the spectrum of the Lax operator $\mathcal{L}_u : H^1_{\rm per} \subset L^2_{\rm per} \to L^2_{\rm per}$ closed in the space of periodic functions with the spatial period of $u$ such that the projection formula (\ref{q-plus-minus-projection}) yields
$$
\oint q^+ dx = -\sigma \oint \bar{u} p dx, \qquad \mu \oint q^- dx = 0.
$$
These requirements are summarized as follows.

\begin{definition}
	\label{def-Lax}
	The Lax spectrum of the linear system (\ref{2.2}) is the set of admissible values of $\lambda$ for which the components $(p,q^+,q^-)$ of the eigenfunction are bounded functions of $x$ with $q^{\pm}$ being analytic and bounded in $\mathbb{C}_{\pm}$ respectively. In addition, if $q^- \neq 0$, then the Lax spectrum includes $\cup_{j=0}^{N} [\lambda_j,\lambda_j + k_1]$, where $\{ \lambda_0, \lambda_1, \dots,\lambda_N \}$ are the admissible values of $\lambda$ for which the mean value of $q^-$ over the spatial period of $u$ is zero.
\end{definition}

The traveling periodic waves, their Lax spectrum, and breathers on their background are considered in Sections \ref{sec-4} and \ref{sec-5}.

\section{Stability of the nonzero constant background}
\label{sec-3}

The simplest solution of the nonlocal model (\ref{INLS}) is the nonzero constant solution $u(x,t) = u_0 \in \mathbb{C} \backslash \{0\}$. The  constant value $u_0 \in \mathbb{C}$ can be normalized to unity without loss of generality due to the scaling symmetry (\ref{NLS-scaling}) and the rotational symmetry (\ref{NLS-symm}). 

The following result gives the linear stability of the nonzero constant background with respect to decaying perturbations.

\begin{theorem}
	\label{prop-stab}
	Let $u = 1 + v$ and consider the linearized equations of motion
	\begin{equation}
	\label{lin-NLS}
	i v_t = v_{xx} + \sigma (i + H) (v_x + \bar{v}_x).
	\end{equation}
If $\sigma = +1$, then for every initial data $v_0 \in H^s(\R)$, $s \geq 0$, the unique solution $v \in C^0(\mathbb{R},H^s(\R))$ to the linearized equation (\ref{lin-NLS}) with $v |_{t=0} = v_0$ satisfies  
\begin{equation}
\label{bound-on-error}
	\| v(\cdot,t) \|_{H^s} \leq C \| v_0 \|_{H^s} \quad \mbox{\rm for every } \;\; t \in \R,
\end{equation}
for some constant $C > 0$. If $\sigma = -1$, then for every $v_0 \in H^s(\R) \cap L^{2,p}(\R)$, $s \geq 0$, $p > \frac{3}{2}$ with $\hat{v}_0 \in C^1(\mathbb{R})$ satisfying $\hat{v}_0(\pm 1) = 0$, the unique solution $v \in C^0(\mathbb{R},H^s(\R))$ to the linearized equation (\ref{lin-NLS}) with $v |_{t=0} = v_0$ satisfies  
\begin{equation}
\label{bound-on-error-focusing}
\| v(\cdot,t) \|_{H^s} \leq C \| v_0 \|_{H^s \cap L^{2,p}} \quad \mbox{\rm for every } \;\; t \in \R.
\end{equation}
Here $\hat{v}_0$ is the Fourier transform of $v_0$ and $L^{2,p}(\R) = \{ f \in L^2(\R) : \;\; |x|^p f \in L^2(\R)\}$.
\end{theorem}

\begin{proof}
	Separating the real and imaginary parts in the linear equation (\ref{lin-NLS}) as $v = A + i B$ yields the coupled system 
	\begin{equation*}
	\left\{ \begin{array}{l} A_t = B_{xx} + 2 \sigma A_x, \\
	-B_t = A_{xx} + 2 \sigma H A_x. \end{array}
	\right.
	\end{equation*}
Let $\hat{A},\hat{B}$ denote the Fourier transform of $A,B$ with Fourier parameter $k \in \mathbb{R}$. The Fourier transform brings the system to the form 
	\begin{equation}
	\label{ode-system}
\left\{ \begin{array}{l} \hat{A}_t = -k^2 \hat{B} + 2 \sigma i k \hat{A}, \\
\hat{B}_t = k^2 \hat{A} + 2 \sigma |k| \hat{A}, \end{array}
\right.
\end{equation}	
from which we obtain the following characteristic equation,
$$
\lambda^2 - 2 i \sigma k \lambda + k^2 (k^2 + 2 \sigma |k| ) = 0.
$$
Due to the factorization
$$
\left( \lambda - i \sigma k \right)^2 + k^2 \left( |k| + \sigma \right)^2 = 0,
$$
the characteristic equation admits two solutions 
\begin{align*}
\lambda_1(k) &= -i k |k|, \\
\lambda_2(k) &= i k (2 \sigma + |k|).
\end{align*}
Since $\lambda_1(k), \lambda_2(k) \in i \R$, solution $\hat{A}$ and $\hat{B}$ of system (\ref{ode-system}) are bounded functions of $t$ if $\lambda_1(k) \neq \lambda_2(k)$. For $k = 0$, we have $\lambda_1(0) = \lambda_2(0) = 0$ but the system (\ref{ode-system}) gives constant $\hat{A}$ and $\hat{B}$ in $t$. 

\underline{If $\sigma = +1$,} then there exist no other solutions of $\lambda_1(k) = \lambda_2(k)$ with $k \neq 0$. Hence, for every $k \in \mathbb{R}$, there exists $C > 0$ such that 
\begin{equation}
\label{estimate-A-B}
|\hat{A}(k,t)| + |\hat{B}(k,t)| \leq C \left( |\hat{A}(k,0)| + |\hat{B}(k,0)| \right), \qquad t \in \mathbb{R},
\end{equation}
so that the bound (\ref{bound-on-error}) holds. 

\underline{If $\sigma = -1$,}  then there exist solutions $k = \pm 1$ of $\lambda_1(k) = \lambda_2(k) = \mp i$. For every $k \neq \pm 1$, we obtain the unique bounded solution of the system (\ref{ode-system}),
\begin{align*}
\hat{A}(k,t) &= k \hat{C}_1(k) e^{-ik|k| t} + \hat{C}_2(k) e^{ik (|k|- 2) t}, \\
\hat{B}(k,t) &= i(|k|-2) \hat{C}_1(k) e^{-ik|k| t} - i {\rm sgn}(k) \hat{C}_2(k) e^{ik (|k|- 2) t},
\end{align*}
for some $t$-independent $\hat{C}_1(k)$ and $\hat{C}_2(k)$ that only depend on $\hat{A}(k,0)$ and $\hat{B}(k,0)$ according to the exact expressions:
\begin{align*}
\hat{C}_1(k) &= \frac{i {\rm sgn}(k) \hat{A}(k,0) + \hat{B}(k,0)}{2i(|k|-1)}, \\
\hat{C}_2(k) &= \frac{i (|k|-2) \hat{A}(k,0) - k \hat{B}(k,0)}{2i(|k|-1)}.
\end{align*}
If $v_0 \in H^s(\R) \cap L^{2,p}(\R)$, $s \geq 0$, $p > \frac{3}{2}$, then  $\hat{A}(k,0), \hat{B}(k,0)$ are $C^1(\mathbb{R})$ functions by the Fourier theory. If they satisfy the constraints $\hat{A}(\pm 1,0) = \hat{B}(\pm 1,0) = 0$, then the estimate (\ref{estimate-A-B}) is replaced by 
\begin{equation*}
|\hat{A}(k,t)| + |\hat{B}(k,t)| \leq C \left\{ \begin{array}{ll} 
|\hat{A}'(k,0)| + |\hat{B}'(k,0)|, \quad & |k| \leq 2, \\
|\hat{A}(k,0)| + |\hat{B}(k,0)|, \quad & |k| > 2, \end{array} \right. 
\qquad t \in \mathbb{R},
\end{equation*}
so that the bound (\ref{bound-on-error-focusing}) holds. 
\end{proof}

\begin{remark}
	In the defocusing case $\sigma = +1$, the linear stability of Theorem \ref{prop-stab} can be extended to the space of periodic functions $H^s_{\rm per}(0,L)$, $s \geq 0$ for every period $L > 0$.
\end{remark}

\begin{remark}
	\label{rem-instability}
	In the focusing case $\sigma = -1$, the resonance of $\lambda_1(k) = \lambda_2(k)$ for $k = \pm 1$ suggests the linear instability of the constant solution $u = 1$ in the space of $2\pi$-periodic functions. Indeed, the system (\ref{ode-system}) for $k = \pm 1$ and $\sigma = -1$ admits two solutions, one of which is linearly growing in $t$:
\begin{align*}
\hat{A}(\pm 1,t) &= (\hat{c}_1  +  \hat{c}_2 t) \; e^{\mp i t}, \\
\hat{B}(\pm 1,t) &= (\mp i \hat{c}_1 + (\mp i t - 1) \hat{c}_2) \; e^{\mp i t},
\end{align*}	
for some $t$-independent $\hat{c}_1$ and $\hat{c}_2$ obtained from $\hat{A}(\pm 1,0)$ and $\hat{B}(\pm 1,0)$. This linear instability is missed in $L^2_{\rm per}(0,L)$ if the spatial period $L$ is not divisible by $2 \pi$.
\end{remark}

In order to obtain the nonlinear stability of the nonzero constant background, we use the conserved quantities, see equations (6.8) in \cite{PG95} 
and (A.16)--(A.18) in \cite{MatsunoSAMP}. The nonlocal model (\ref{INLS}) on $\mathbb{R}$ with the boundary conditions $|u(x,t)| \to 1$ as $|x| \to \infty$ for every $t \in \mathbb{R}$ admits the following conserved quantities:
	\begin{align*}
	I_1(u) &= \int_{\R} (|u|^2 - 1) dx, \\
	I_2(u) &= i \int_{\R} (u\bar{u}_x - \bar{u} u_x) dx + \sigma \int_{\R} (|u|^4 - 1) dx, \\
	I_3(u) &= \int_{\R} \left( |u_x|^2 - \frac{i}{2} \sigma |u|^2 (\bar{u} u_x -  \bar{u}_x u) - \frac{1}{2} \sigma |u|^2 H(|u|^2)_x + \frac{1}{3} (|u|^6 - 1) \right) dx.
	\end{align*}
A suitable combination of the conserved quantities leads to the following nonlinear stability result in the defocusing case $\sigma = +1$.	

\begin{theorem}
	\label{theorem-stab}
For every fixed $L > 0$, there exists $\delta > 0$ such that for every $v_0 \in H^1_{\rm per}((0,L),\mathbb{C})$ with $\| v_0 \|_{H^1_{\rm per}} \leq \delta$, the unique solution $u \in C^0(\mathbb{R},H^1_{\rm per}((0,L),\mathbb{C}))$ to the nonlocal model (\ref{INLS}) with $\sigma = +1$ and with $u |_{t=0} = 1 + v_0$ satisfies 
\begin{equation}
\label{bound-cont}
\| e^{-i \theta(t)} u(\cdot,t) - 1 \|_{H^1_{\rm per}} \leq C \| v_0 \|_{H^1_{\rm per}} \quad \mbox{\rm for every } \;\; t \in \R,
\end{equation}
for some constant $C > 0$ and some function $\theta \in C^0(\mathbb{R})$.
\end{theorem}	
	
\begin{proof}
	By substituting $u = 1 + v$ in the conserved quantities, expanding them in $v$, and integrating by parts, we obtain 
	\begin{align*}
	I_1(1+v) &= \int (v + \bar{v}+ |v|^2) dx, \\
	I_2(1+v) - 2 \sigma I_1(1+v) &= i \int (v \bar{v}_x - \bar{v} v_x) dx + \sigma \int (v + \bar{v}+ |v|^2)^2 dx, \\ 
I_3(1+v) - \sigma I_2(1+v) + I_1(1+v) &= \int \left( |v_x|^2 + \frac{1}{2} \sigma (v + \bar{v}) K (v + \bar{v}) + N(v) \right) dx,
	\end{align*}
	where $K = - H \partial_x$ and $N(v)$ contains nonlinear terms from cubic to  sixth-order powers of $v$:
	\begin{align*}
	N(v) &= i \sigma |v|^2 (\bar{v}_x - v_x) + \frac{i}{2} \sigma (v^2 \bar{v}_x - \bar{v}^2 v_x) + \frac{i}{2} \sigma |v|^2 (v \bar{v}_x - \bar{v} v_x) \\
	& \quad - \sigma (v + \bar{v}) H(|v|^2)_x - \frac{1}{2} \sigma |v|^2 H(|v|^2)_x + \frac{1}{3} (v + \bar{v} + |v|^2)^3.
	\end{align*}
	The integration interval can be considered on the period $(0,L)$ if $v \in H^1_{\rm per}((0,L),\mathbb{C})$ with any fixed period $L > 0$. Local well-posedness of the nonlocal model (\ref{INLS}) in $H^1(\mathbb{R},\mathbb{C})$ has been proven in \cite{Pilod}, this result can be extended in $H^1_{\rm per}((0,L),\mathbb{C})$ \cite{B-22}. 
	
The Lyapunov functional is defined by 
\begin{equation}
\label{Lyp}
\Lambda(v) := I_3(1+v) - \sigma I_2(1+v) + I_1(1+v), \quad v \in H^1_{\rm per}((0,L),\mathbb{C}).
\end{equation} 
In the defocusing case $\sigma = +1$, the quadratic part of $\Lambda$ is positive and coercive in $v \in H^1_{\rm per}((0,L),\mathbb{C})$ for the spatially varying part of the perturbation $v$. Indeed, if we use Fourier series 
$$
v(x) = \sum_{n \in \mathbb{Z}} \hat{v}_n e^{\frac{2\pi i n x}{L}}, \quad 
\bar{v}(x) = \sum_{n \in \mathbb{Z}} \bar{\hat{v}}_{-n} e^{\frac{2\pi i n x}{L}},
$$
then we obtain by Parseval's equality 
\begin{align*}
\oint \left[  |v_x|^2 + \frac{1}{2} (v + \bar{v}) K (v + \bar{v}) \right] dx 
&= \sum_{n \in \mathbb{Z}} \frac{4 \pi^2 n^2}{L} |\hat{v}_n|^2 + \pi |n| |\hat{v}_n + \bar{\hat{v}}_{-n}|^2,
\end{align*}
where we have used the Fourier symbol of $K$ from $K(e^{ikx}) = |k| e^{ikx}$, $k \in \mathbb{R}$.
Neglecting the second term in the lower bound and using Poincar\'{e} inequality for the first term, we get the coercivity bound 
\begin{align}
\label{bound-1}
\oint \left[  |v_x|^2 + \frac{1}{2} (v + \bar{v}) K (v + \bar{v}) \right] dx 
\geq \frac{1}{2} \| v_x \|_{L^2}^2 + \frac{2\pi^2}{L^2} \| v - \hat{v}_0 \|_{L^2}^2,
\end{align}
which allows us to control the $H^1_{\rm per}((0,L),\mathbb{C})$ norm of the spatially varying part of the local solution 
$v \in C^0((-\tau_0,\tau_0),H^1_{\rm per}((0,L),\mathbb{C}))$ 
for some $\tau_0 > 0$ from the conserved value of the Lyapunov functional $\Lambda$.

It remains to control the mean value of the perturbation $v$. We can preserve the zero-mean constraint for the imaginary part of  $v$ by using the rotational invariance (\ref{NLS-symm}) and introducing the orthogonal decomposition 
$$
u(x,t) = e^{i \theta(t)} \left[ 1 + v(x,t) \right], \qquad \oint {\rm Im}(v) dx = 0,
$$
where the modulational parameter $\theta \in C^0((-\tau_0,\tau_0),\R)$ is uniquely defined for the local solution $v \in C^0((-\tau_0,\tau_0),H^1_{\rm per}((0,L),\mathbb{C}))$ from zeros of $f(\theta) : \R \to \R$ given by
$$
f(\theta) := \oint {\rm Im}(e^{-i \theta} u - 1) dx. 
$$
By the implicit function theorem, there exists a unique $\theta \in \R$ for every $u \in H^1_{\rm per}$ in the ball with small $\inf_{\theta \in \R} \| e^{-i\theta} u - 1 \|_{H^1_{\rm per}} \leq C \| v_0 \|_{H^1_{\rm per}}$.

To control the mean value of the real part of the perturbation $v$, we use the first conserved quantity $I_1$ and Parseval's equality to obtain 
$$
I_1(1+v) = 2L \hat{v}_0 + L \hat{v}_0^2 + L \sum_{n \in \mathbb{Z} \backslash \{0\}} |\hat{v}_n|^2, 
$$
where $\hat{v}_0 \in \R$ due to the zero-mean constraint $\oint {\rm Im}(v) dx = 0$. This  yields 
\begin{align*}
L (\hat{v}_0 + 1)^2 &= I_1(1+v) + L - L \sum_{n \in \mathbb{Z} \backslash\{0\}} |\hat{v}_n|^2 \leq I_1(1+v) + L,
\end{align*}
or 
$$
|\hat{v}_0| \leq \frac{\sqrt{L + I_1(1+v)}}{\sqrt{L}} - 1. 
$$ 
Since $I_1(1+v)$ is conserved in time, Cauchy--Schwarz inequality implies that 
$$
|I_1(1+v)| \leq 2 \sqrt{L} \| v_0 \|_{L^2} + \| v_0 \|^2_{L^2}.
$$
Due to the smallness of $\| v_0 \|_{L^2}$, there is a constant $C > 0$ independently of the initial data $v_0 \in H^1_{\rm per}((0,L),\mathbb{C})$ (which may change from one line to another line) such that 
\begin{equation}
\label{bound-2}
|\hat{v}_0(t)| \leq \frac{\sqrt{L + I_1(1+v)}}{\sqrt{L}} - 1 \leq C \| v_0 \|_{L^2},
\end{equation}
which controls $\hat{v}_0(t)$ for every $t \in (-\tau_0,\tau_0)$.

Due to the Banach algebra of $H^1_{\rm per}$, the nonlinear terms of $\Lambda(v)$ are controlled by 
\begin{equation}
\label{bound-3}
\left| \int N(v) dx \right| \leq C \left( \| v \|_{H^1_{\rm per}}^3 + \| v \|_{H^1_{\rm per}}^6 \right).
\end{equation} 
Since the value of $\Lambda(v)$ is conserved in time $t \in \R$ and $\|v_0 \|_{H^1_{\rm per}}$ is small, we obtain the bound
$$
\Lambda(v) \leq C  \| v_0 \|_{H^1_{\rm per}}^2.
$$
By using the coercivity of the quadratic part of $\Lambda(v)$ for the varying part of $v$, the zero-mean constraint for ${\rm Im}(v)$, and the control of the mean value of ${\rm Re}(v)$, we obtain with triangle inequality 
and bounds (\ref{bound-1}), (\ref{bound-2}), and (\ref{bound-3}) that 
\begin{align*}
\| v(\cdot,t) \|_{H^1_{\rm per}} & \leq \| \hat{v}_0(t) \|_{L^2} 
+ \| v(\cdot,t) - \hat{v}_0(t) \|_{H^1_{\rm per}} \\
& \leq \sqrt{L} |\hat{v}_0(t)| + C \sqrt{\Lambda(v)} \\
& \leq C \| v_0 \|_{H^1_{\rm per}}
\end{align*}
for every $t \in (-\tau_0,\tau_0)$.  Since this bound is independent of $t$, the local solution $v \in C^0((-\tau_0,\tau_0),H^1_{\rm per}((0,L),\mathbb{C}))$ can be extended globally to yield the bound (\ref{bound-cont}).
\end{proof}
 
\begin{corollary}
	\label{cor-stab}
Theorem \ref{theorem-stab} holds in the focusing case $\sigma = -1$ for every $L \in (0,\pi)$.
\end{corollary}

\begin{proof}
In the focusing case $\sigma = -1$, the quadratic part of $\Lambda(v)$ is 
given by 
\begin{align*}
\oint \left[  |v_x|^2 - \frac{1}{2} (v + \bar{v}) K (v + \bar{v}) \right] dx 
&= \sum_{n \in \mathbb{Z}} \frac{4 \pi^2 n^2}{L} (|\widehat{{\rm Re}(v)}_n|^2 + |\widehat{{\rm Im}(v)}_n|^2) - 4 \pi |n| |\widehat{{\rm Re}(v)}|^2.
\end{align*}
It is clear that it is sign-definite for $L < \pi$, hence the same Lyapunov functional (\ref{Lyp}) can be used for the proof of nonlinear stability of the constant solution $u = 1$ if $L < \pi$. The rest of the proof holds verbatim.
\end{proof}

\begin{remark}
	For $L \in [\pi,\infty)$, it is an open problem to prove the nonlinear stability of the constant solution $u = 1$ with respect to perturbations in $H^1_{\rm per}((0,L),\mathbb{C})$ in the focusing case $\sigma = -1$. This interval includes the periods $L$ multiple to $2\pi$, for which the linear instability holds by Remark \ref{rem-instability}.
\end{remark}

\section{Traveling periodic waves and breathers on the nonzero background}
\label{sec-4} 

We introduce the bilinear formulation of the nonlocal model (\ref{INLS}) 
and the linear system (\ref{2.2}). We use the bilinear formulation to obtain the traveling periodic wave in Section \ref{sec-periodic}. Lax spectrum of the traveling periodic wave is computed in Section \ref{sec-lax}. 
By using the Lax spectrum and the $2$-periodic solutions, we construct the exact solution for the solitary wave on the background of the traveling periodic wave in Section \ref{sec-breathers}. 
Section \ref{sec-n-breathers} gives a closed-form solution for $N$ solitary waves on the background of the traveling periodic wave as a quotient of determinants. 

Without loss of generality, we normalize the nonzero background for the traveling periodic wave to unity due to the scaling symmetry (\ref{NLS-scaling}). We also refer to solitary waves on the traveling periodic wave as to the breathers, similar 
to the terminology used in \cite{ChenPel24,HMP,MP-24}, due to the periodic character of the interaction between the solitary wave and the traveling periodic wave.

\subsection{Traveling periodic wave} 
\label{sec-periodic}

Assume that $f$ and $\tilde{f}$ have only zeros in the lower and upper half of the complex plane of $x$, respectively. Then, $f_x/f$ and $\tilde{f}_x/\tilde{f}$ are analytic in the upper and lower half-planes, respectively. By using the projection formulas (\ref{projection-formulas}), we obtain 
\begin{equation}
\label{2.4}
\mathcal{P}^+ \frac{\partial^2}{\partial x^2} \ln \frac{f}{\tilde{f}} =   \frac{\partial^2}{\partial x^2} \ln f, \quad \mathcal{P}^- \frac{\partial^2}{\partial x^2} \ln \frac{f}{\tilde{f}} = - \frac{\partial^2}{\partial x^2} \ln \tilde{f}.
\end{equation}
Substitution
\begin{equation}
\label{2.3}
u = \frac{g}{f}, \quad \bar{u} = \frac{\tilde{g}}{\tilde{f}}, \quad |u|^2 = 1 - i \sigma \frac{\partial}{\partial x} \ln \frac{f}{\tilde{f}}
\end{equation}
transforms the nonlocal model \eqref{INLS} into the following system of bilinear equations:
\begin{equation}
\label{2.5}
\left\{ 
\begin{array}{l}
(i D_t + D_x^2) f \cdot g = 0, \\
(-i D_t + D_x^2) \tilde{f} \cdot \tilde{g} = 0, \\
i D_x f \cdot \tilde{f} +  \sigma ( g \cdot \tilde{g} -  f \cdot \tilde{f} ) = 0.   
\end{array}  \right.
\end{equation}

The following proposition summarizes the state-of-art in the existence of the traveling periodic waves on the nonzero background. Although we give a proof for the sake of completeness, similar solution waveforms have been obtained in \cite{B-23,MatsunoPLA,MatsunoSAMP}.

\begin{proposition}
	\label{prop-1}
The nonlocal model (\ref{INLS}) admits the traveling periodic wave in the form
	\begin{equation}\label{2.10}
	u(x,t) = e^{\frac{1}{2} (\psi_1 - \phi_1)} \frac{1 + e^{i k_1 \xi_1 - \psi_1}}{1 + e^{i k_1 \xi_1 - \phi_1}}, \quad 
	\bar{u}(x,t) = e^{-\frac{1}{2} (\psi_1 - \phi_1)} \frac{1 + e^{i k_1 \xi_1 + \psi_1}}{1 + e^{i k_1 \xi_1 + \phi_1}}
\end{equation}
and 
\begin{equation}\label{2.11}
|u(x,t)|^2 = 1 - \frac{\sigma k_1 \sinh \phi_1}{\cos k_1 \xi_1 + \cosh \phi_1},
\end{equation}	
where $k_1 > 0$ and $\xi_1 = x - c_1 t - x_1$ with arbitrary $x_1 \in \mathbb{R}$, whereas $\phi_1 > 0$ and $\psi_1 \in \mathbb{R}$ are uniquely determined by 
\begin{equation}
\label{2.9} 
e^{2\phi_1} = \frac{(c_1 - k_1) (c_1 + k_1 + 2 \sigma)}{(c_1 + k_1) (c_1 - k_1 + 2 \sigma)}, \qquad e^{\psi_1} = \frac{c_1 + k_1}{c_1 - k_1} e^{\phi_1}.
\end{equation}
The parameters $k_1 > 0$ and $c_1 \in \mathbb{R}$ are further restricted as follows:
\begin{itemize}
	\item If $\sigma = 1$, then $k_1 \in (0,1)$ and $c_1 \in (-2+k_1,-k_1)$. 
	\item If $\sigma = -1$, then $k_1 \in (0,\infty)$ and either $c_1 \in (k_1+2,\infty)$ or $c_1 \in (-\infty,-k_1)$.
\end{itemize}
\end{proposition}

\begin{proof}
Let us consider the following $1$-periodic solution of the bilinear equations (\ref{2.5}):
\begin{equation}
	\label{2.8} 
	\left\{ 
\begin{array}{cc}
f = 1 + e^{i k_1 \xi_1 - \phi_1}, &
\tilde{f} = 1 + e^{i k_1 \xi_1 + \phi_1}, \\
g = \gamma_1 (1 + e^{i k_1 \xi_1 - \psi_1}), & \tilde{g} = 
\gamma_1^{-1} (1 + e^{i k_1 \xi_1 + \psi_1}), 
\end{array}
\right.
\end{equation}
where $\xi_1 = x - c_1 t - x_1$. The real parameters $k_1$, $c_1$, $x_1$ 
are arbitrary as long as the sign of $\phi_1$ coincides with the sign of $k_1$, 
since $f$ and $\tilde{f}$ given by (\ref{2.8}) must only have zeros in the lower and upper half-planes, respectively. The real parameter $\gamma_1$ is arbitrary as long as $\bar{u}$ is a complex conjugate of $u$.

The bilinear equations (\ref{2.5}) are satisfied if and only if coefficients $\phi_1$ and $\psi_1$ are uniquely determined by (\ref{2.9}). Substituting \eqref{2.8} into \eqref{2.3}, we obtain the traveling wave solution in the form 
(\ref{2.10}) and (\ref{2.11}) with $\bar{u}$ being the complex conjugate of $u$ if and only if $\gamma_1$ is a real root of the quadratic equation:
$$
\gamma_1^2 = e^{\psi_1 - \phi_1} = \frac{c_1 + k_1}{c_1 - k_1}. 
$$

It remains to obtain the admissible values for parameters $k_1$ and $c_1$ from the condition that the sign of $\phi_1$ must coincide with the sign of $k_1$ 
and that $\gamma_1^2 > 0$. Due to the symmetry of $k_1$ in the expression for $e^{2 \phi_1}$ given by (\ref{2.9}), we can consider $k_1 > 0$ without loss of generality, with $\phi_1 > 0$. The admissible values for $k_1$ and $c_1$ are defined from the inequalities:
\begin{equation}
\label{inequality}
\frac{c+k_1}{c-k_1} > 0 \quad \mbox{\rm and} \quad \frac{(c_1 - k_1) (c_1 + k_1 + 2 \sigma)}{(c_1 + k_1) (c_1 - k_1 + 2 \sigma)} > 1.
\end{equation}
There are four cases to be considered for both $\sigma = 1$ and $\sigma = -1$. \\

\underline{If $\sigma = +1$, we obtain:}
\begin{itemize}
	\item If $c_1 - k_1 > 0$, the second inequality in (\ref{inequality}) yields $k_1 < 0$, a contradiction. 
	
	\item If $c_1 - k_1 < 0$ and $c_1 + k_1 > 0$, then the first inequality in (\ref{inequality}) is contradictory.	
	
	\item If $c_1 - k_1 < 0$ and $c_1 + k_1 < 0$ but $c_1 - k_1 + 2 > 0$ and $c_1 + k_1 + 2 > 0$, then $k_1 \in (0,1)$ and $c_1 \in (-2+k_1,-k_1)$. The second inequality in (\ref{inequality}) yields $k_1 > 0$, which is true. This is the only case for the family of periodic traveling wave.
	
	\item If $c_1 - k_1 < 0$ and $c_1 + k_1 < 0$ but $c_1 - k_1 + 2 < 0$ and $c_1 + k_1 + 2 < 0$, the second inequality in (\ref{inequality}) yields $k_1 < 0$, a contradiction.  
\end{itemize}
\vspace{0.25cm}

\underline{If $\sigma = -1$, we obtain:}
\begin{itemize}
	\item If $c_1 > k_1 + 2$, both inequalities in (\ref{inequality}) are satisfied for $k_1 > 0$. 
	
	\item If $c_1 > k_1$ but $c_1 < 2 - k_1$, then $k_1 \in (0,1)$ but the second  inequality in (\ref{inequality}) yields $k_1 < 0$, a contradiction.
	
	\item If $c_1 < k_1$ but $c_1 > -k_1$, the first inequality in (\ref{inequality}) is contradictory.
	
	\item If $c_1 < -k_1$, both inequalities in (\ref{inequality}) are satisfied for $k_1 > 0$. 
\end{itemize} 

Thus, only one solution exists for $\sigma = +1$ with $k_1 \in (0,1)$ and $c_1 \in (-2+k_1,-k_1)$ and two solutions exist for $\sigma = -1$ with either $c_1 \in (k_1 + 2,\infty)$ or with $c_1 \in (-\infty,-k_1)$ for every 
$k_1 > 0$.
\end{proof}

\begin{remark}
	\label{rem-analytic}
	It follows from (\ref{2.10}) that $u$ and $\bar{u}$ are analytic in $\mathbb{C}_+$ and $\mathbb{C}_-$ respectively. This was not a requirement on solutions of the nonlocal model (\ref{INLS}). Nevertheless, the traveling periodic waves satisfy this property. 
\end{remark}

\begin{remark}
	The traveling periodic waves can be extended by using the translational and rotational symmetries (\ref{NLS-symm}) and the Lorentz transformation (\ref{NLS-Lorentz}), whereas the scaling transformation (\ref{NLS-scaling}) has been used to normalize the nonzero background to unity.
\end{remark}

To study properties of the traveling periodic waves of Proposition \ref{prop-1}, we note that the existence intervals for the wave speed $c_1$ are symmetric relative to $-\sigma$. In other words, replacing $c_1$ by $-2\sigma - c_1$  yields the same expression for $e^{2 \phi_1}$ in (\ref{2.9}) and hence for $|u(x,t)|^2$ in (\ref{2.11}). The existence intervals can be formulated symmetrically as 
\begin{equation}
\label{existence-1}
\sigma = +1 : \qquad c_1 + 1 \in (-1+k_1,1-k_1), \qquad k_1 \in (0,1)
\end{equation}
and 
\begin{equation}
\label{existence-2}
\sigma = -1 : \qquad c_1 - 1 \in (-\infty,-1-k_1) \cup (1+k_1,\infty), \qquad k_1 \in (0,\infty).
\end{equation}
Without loss of generality, we can consider travelng periodic waves for $c_1 + \sigma \geq 0$. 

\begin{figure}[htb!]
	\centering
	\includegraphics[width=7.5cm,height=7cm]{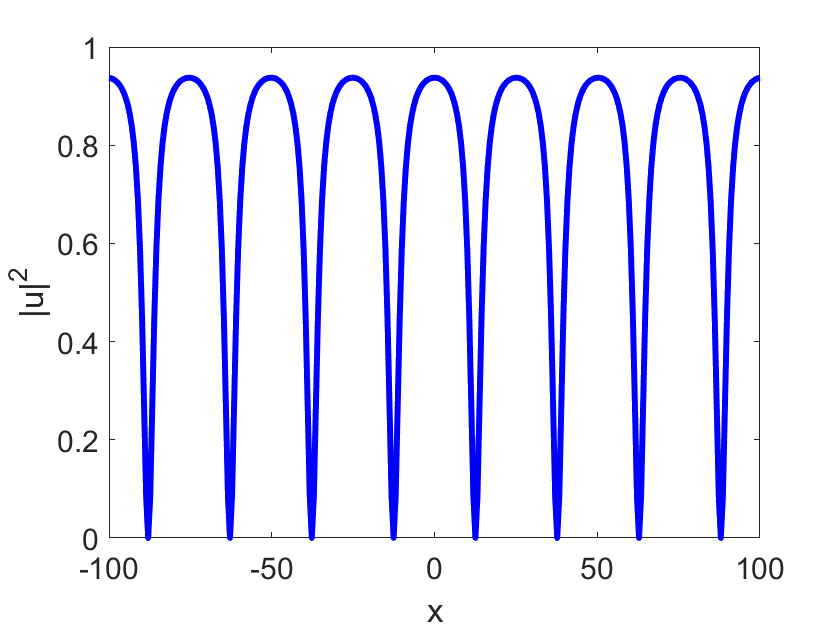}
	\includegraphics[width=7.5cm,height=7cm]{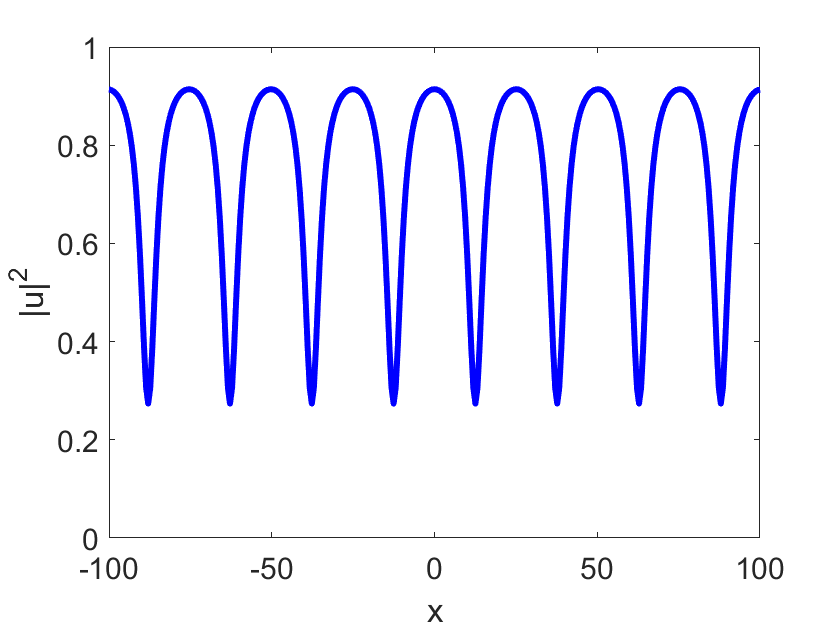}
	\caption{The profile of $|u|^2$ versus $x$ for $\sigma = +1$, $k_1 = 0.25$, and either $c_1 = -1$ (left) or $c_1 = -0.5$ (right).}
	\label{fig-1}
\end{figure}

\begin{figure}[htb!]
	\centering
	\includegraphics[width=7.5cm,height=7cm]{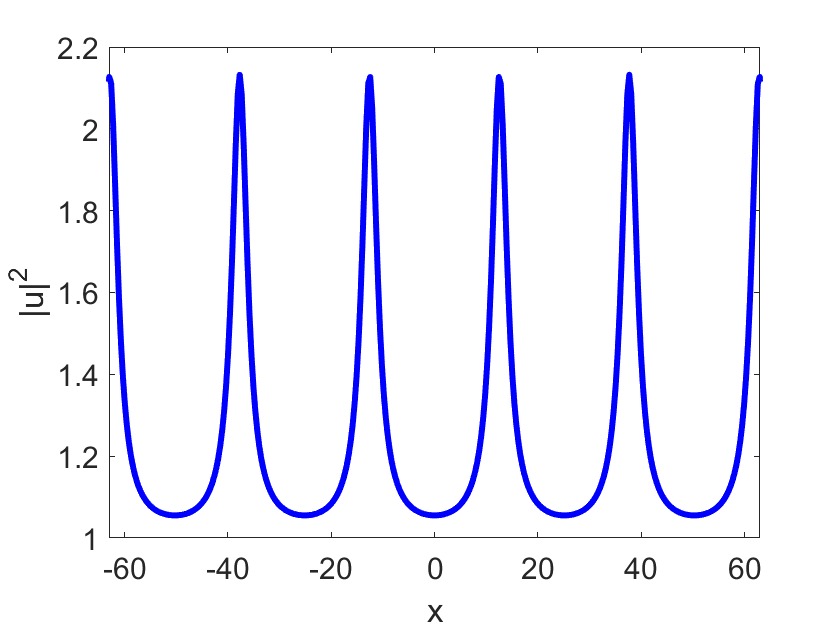}
	\includegraphics[width=7.5cm,height=7cm]{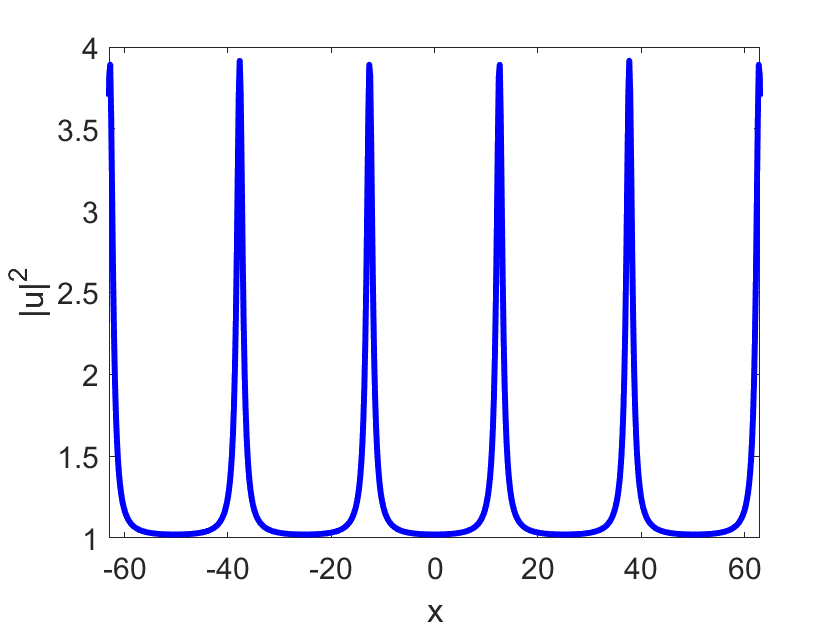}
	\caption{The profile of $|u|^2$ versus $x$ for $\sigma = -1$, $k_1 = 0.25$, and either $c_1 = 2 + 2k_1$ (left) or $c_1 = 2 + 4k_1$ (right).}
	\label{fig-1-focus}
\end{figure}

Figure \ref{fig-1} gives the spatial profile of $|u|^2$ given by (\ref{2.11}) for $\sigma = +1$, $k_1 = 0.25$, and two choices of $c_1$ in (\ref{existence-1}). The wave profiles are of the depression type and 
the profile with $c_1 = -1$ reaches the zero value (left panel). Figure \ref{fig-1-focus} shows the spatial profile of $|u|^2$ for $\sigma = -1$, $k_1 = 0.25$, and two choices of $c_1$ in (\ref{existence-2}). The wave profiles are of the elevation type with larger amplitudes for larger speeds.

The long-wave limit of the periodic wave appears as $k_1 \to 0$ where the periodic wave reduces to a solitary wave. As $k_1 \rightarrow 0$, it follows from \eqref{2.9} that $\phi_1 \rightarrow 0$ and $\psi_1 \rightarrow 0$ according to the power expansions
$$
\phi_1 = -\frac{2 \sigma k_1}{c_1 (c_1+2 \sigma)} + \mathcal{O}(k_1^2), \qquad 
\psi_1 = \frac{2(c_1+\sigma) k_1}{c_1 (c_1+2 \sigma)} + \mathcal{O}(k_1^2),
$$
For each family of Proposition \ref{prop-1}, we have $\phi_1 > 0$ if $k_1 > 0$.
We obtain from (\ref{2.10}) and (\ref{2.11}) after the transformation $x_1 \mapsto x_1  + \pi/k_1$ in the limit $k_1 \to 0$ that 
\begin{equation} \label{2.12} 
u(x,t) = 1 - \frac{2 \sigma (c_1+2 \sigma)}{2 + i \sigma c_1 (c_1 + 2 \sigma ) \xi_1}  
\end{equation}
and
\begin{equation} \label{2.13} 
|u(x,t)|^2 = 1 + \frac{4 c_1 (c_1+2 \sigma)}{c_1^2 (c_1 + 2 \sigma )^2 \xi^2_1 + 4}.  
\end{equation}

\underline{If $\sigma = +1$,} the existence interval for $c_1 \in (-2+k_1,-k_1)$ becomes 
$c_1 \in (-2,0)$ in the limit $k_1 \to 0$. Since $c_1 < 0$ and $c_1 + 2 > 0$, the algebraic soliton is the dark soliton on the nonzero constant  background with the smallest intensity attained at $\xi_1 = 0$:
$$
\min\limits_{(x,t) \in \R^2} |u(x,t)|^2 = (1+c_1)^2 < 1.
$$

\underline{If $\sigma = -1$,} the existence intervals for  $c_1 \in (k_1 + 2,\infty)$ and $c_1 \in (-\infty,-k_1)$ become $c_1 \in (2,\infty)$ and $c_1 \in (-\infty,0)$ in the limit $k_1 \to 0$. Since either $c_1 < 0$ or $c_1 > 2$, the algebraic soliton is the bright soliton on the nonzero constant background with the largest intensity attained at $\xi_1 = 0$: 
$$
\max\limits_{(x,t) \in \R^2} |u(x,t)|^2  = (1-c_1)^2 > 1.
$$

\subsection{Lax spectrum of the traveling periodic wave}
\label{sec-lax}

To obtain the exact solutions of the linear system (\ref{2.2}), we use the representation (\ref{2.3}) and introduce 
\begin{equation} 
\label{eig-bilinear}
p = \frac{\varphi}{f}, \qquad q^+ = \frac{h}{f}, \qquad 
q^- = \frac{\tilde{h}}{\tilde{f}},
\end{equation}
where $f$ and $\tilde{f}$ are given by (\ref{2.8}) and $\varphi$, $h$, and $\tilde{h}$ are to be found. By using (\ref{2.4}), (\ref{2.3}), and (\ref{eig-bilinear}), the linear system (\ref{2.2}) is reduced to the system of bilinear equations:
\begin{equation}
\left\{ 
\begin{array}{l}
(i D_x + \lambda) \varphi \cdot f + g \cdot h  = 0, \\
h \cdot \tilde{f} - \mu \tilde{h} \cdot f + \sigma \varphi \cdot \tilde{g} = 0, \\
(i D_t + \lambda^2) \varphi \cdot f + (i D_x + \lambda) h \cdot g = 0, \\  
(i D_t - 2 i \lambda D_x + D_x^2) h \cdot f = 0, \\   
(i D_t - 2 i \lambda D_x + D_x^2) \tilde{h} \cdot \tilde{f} = 0.
\end{array}  \right.
\label{1.19}
\end{equation} 

The following proposition identifies the Lax spectrum 
for the traveling periodic wave with the spatial profile (\ref{2.10}) and (\ref{2.11}) based on the exact solutions of the system (\ref{1.19}).

\begin{proposition}
	\label{prop-2} 
	Let $u$ be the traveling periodic wave in Proposition \ref{prop-1}. The Lax spectrum in Definition \ref{def-Lax} is located in 
	\begin{equation}
	\label{Lax-spectrum}
	\Sigma = [\lambda_0,\lambda_0+k_1] \cup [\sigma,\infty), \qquad \lambda_0 := -\frac{c_1+k_1}{2}, \quad \lambda_0+k_1 = -\frac{c_1 - k_1}{2}.
	\end{equation}
	\begin{itemize}
		\item If $\sigma = +1$, then $k_1 \in (0,1)$ and $c_1 \in (-2+k_1,-k_1)$ so that 
		$[\lambda_0,\lambda_0+k_1] \subset (0,1)$ is isolated from $[1,\infty)$.
		
		\item If $\sigma = -1$, then $k_1 \in (0,\infty)$ and either $c_1 \in (k_1+2,\infty)$ for which $[\lambda_0,\lambda_0+k_1]$ is isolated from $[-1,\infty)$ or $c_1 \in (-\infty,-k_1)$ for which $[\lambda_0,\lambda_0+k_1]$ is embedded into $[-1,\infty)$.
	\end{itemize}
\end{proposition}

\begin{proof}
We proceed differently for $q^- \equiv 0$ and $q^- \neq 0$. \\

\underline{If $q^- \equiv 0$}, then $\tilde{h} \equiv 0$. The second equation of system (\ref{1.19}) implies that 
$$
h = - \frac{\sigma \varphi \tilde{g}}{\tilde{f}}.
$$
Since $q^+ = -\sigma \bar{u} p$ is analytic in $\mathbb{C}_+$, then $h$ is required to be analytic in $\mathbb{C}_+$. Since $\tilde{f}$ admits zeros in $\mathbb{C}_+$, then $\varphi$ must be divisible by $\tilde{f}$ so that 
\begin{equation} 
\label{1.14}
\varphi = m \tilde{f}, \qquad h = - \sigma m \tilde{g},
\end{equation}
with some $m = m(x,t)$ to be determined (required to be analytic in $\mathbb{C}_+$). 

From the first equation of 
system (\ref{1.19}) we find with the help of the third equaton of system (\ref{2.5}) that 
\begin{equation} 
\label{m-eq}
i m_x + (\lambda - \sigma) m = 0.
\end{equation}
From the third equation of system (\ref{1.19}), we obtain with the help of the third equaton of system (\ref{2.5}) and (\ref{m-eq}) that
$$
\left[ i m_t + (\lambda^2-1) m \right] \tilde{f} \cdot f + i m \left[ (D_t - \sigma D_x) \tilde{f}\cdot f - \sigma  D_x \tilde{g} \cdot g  \right] = 0, 
$$
From the fourth equation of system (\ref{1.19}), we obtain with the help of (\ref{m-eq}) that 
$$
\left[ i m_t + (\lambda^2 -1) m \right]  \tilde{g} \cdot f 
+ m (i D_t - 2 i \sigma D_x + D_x^2) \tilde{g}\cdot f = 0.
$$
By using the exact solution (\ref{2.9}) and (\ref{2.8}), we verify 
that 
\begin{align*}
(D_t - \sigma D_x) \tilde{f}\cdot f - \sigma D_x \tilde{g} \cdot g &= 0, \\
(i D_t - 2 i \sigma D_x + D_x^2) \tilde{g}\cdot f &= 0,
\end{align*}
which imply that 
\begin{equation}
\label{m-frak-eq}
i m_t + (\lambda^2 -1) m = 0.
\end{equation}
Solving (\ref{m-eq}) and (\ref{m-frak-eq}) yields 
$$
m(x,t) = e^{i(\lambda - \sigma) x + i(\lambda^2 - 1) t}
$$ 
with the constant of integration normalized to unity.
By using (\ref{eig-bilinear}) and (\ref{1.14}), we obtain the exact expression for the components $p$ and $q^+$ of the eigenfunctions with $q^- \equiv 0$:
\begin{equation}
\label{1.17}
p = e^{i(\lambda - \sigma)x + i (\lambda^2 - 1) t } \frac{1 + e^{i k_1 \xi_1 + \phi_1}}{1 + e^{i k_1 \xi_1 - \phi_1}}, \qquad q^+ = -\sigma \gamma_1^{-1} e^{i(\lambda - \sigma)x + i (\lambda^2 - 1) t } \frac{1 + e^{i k_1 \xi_1 + \psi_1}}{1 + e^{i k_1 \xi_1 - \phi_1}}.
\end{equation}
The component $q^+$ is analytic in $\mathbb{C}_+$ and bounded as ${\rm Im}(x) \to +\infty$ for every $t \in \R$ if and only if 
$\lambda \geq \sigma$. Hence, $[\sigma,\infty) \in \Sigma$ belongs to the Lax spectrum (\ref{Lax-spectrum}). \\

\underline{If $q^- \neq 0$}, then we obtain solutions for $h$ and $\tilde{h}$ by using the last two equations of system \eqref{1.19}. Given $f$ in (\ref{2.8}) we separate the variables in the form
\begin{equation*}
h = e^{i (\theta \xi_1 + \Omega t)} \left(1 + A e^{i k_1 \xi_1 - \phi_1}\right),
\end{equation*}
with some $\theta$, $\Omega$, and $A$ to be determined. The fourth equation in system \eqref{1.19} is satisfied if and only if 
\begin{equation*}
\Omega = \theta (c_1 + 2 \lambda - \theta) \quad \mbox{\rm and} \quad A  = \frac{c_1 + 2 \lambda  - 2 \theta +  k_1}{c_1 + 2 \lambda  - 2\theta -  k_1},
\end{equation*}
which yields the explicit solution 
\begin{equation} 
\label{1.24}
h = e^{i \theta (\xi_1 + (c_1 + 2 \lambda - \theta) t)} \left(1 + \frac{c_1 + 2 \lambda - 2 \theta +  k_1}{c_1 + 2 \lambda - 2 \theta -  k_1}  e^{i k_1 \xi_1 - \phi_1}\right).
\end{equation}
With similar computations from the fifth equation in system \eqref{1.19}, we obtain the explicit solution 
\begin{equation*} 
%\label{1.26}
	\tilde{h} = e^{i \theta (\xi_1 + (c_1 + 2 \lambda - \theta) t)} \left(1 + \frac{c_1 + 2 \lambda  - 2 \theta +  k_1}{c_1 + 2 \lambda - 2 \theta -  k_1}  e^{i k_1 \xi_1 + \phi_1}\right),
\end{equation*}
where parameter $\theta \in \mathbb{R}$ has to be the same due to the coupling between $h$ and $\tilde{h}$ in the second equation of system (\ref{1.19}). Now $q^+$ and $q^-$ are analytic and bounded in $\mathbb{C}_+$ and $\mathbb{C}_-$ respectively if and only if $\theta = 0$. This yields the unique representation of the components $q^+$ and $q^-$ in the form 
\begin{equation}
\label{q-plus-expression}
q^+ = \frac{1}{1 + e^{i k_1 \xi_1 - \phi_1}} \left[ 1 + \frac{c_1 + 2 \lambda +  k_1}{c_1 + 2 \lambda -  k_1}  e^{i k_1 \xi_1 - \phi_1} \right]
\end{equation}
and 
\begin{equation}
\label{q-minus-expression}
q^- = \frac{1}{1 + e^{i k_1 \xi_1 + \phi_1}} \left[ 1 + \frac{c_1 + 2 \lambda +  k_1}{c_1 + 2 \lambda -  k_1}  e^{i k_1 \xi_1 + \phi_1} \right]. 
\end{equation}
It remains to find $p$ from the first three equations of system (\ref{1.19}). 
Given $f$ and $g$ in (\ref{2.8}) and $h$ in (\ref{1.24}) with $\theta = 0$, 
we separate the variables in the form:
\begin{equation}
\varphi = B  \left( 1 + C e^{i k_1 \xi_1 - \psi_1} \right),
\label{1.35}
\end{equation}
with some parameters $B$ and $C$ to be determined. 
The first equation of system (\ref{1.19}) is satisfied if and only if
$$
B = -\gamma_1 \lambda^{-1} \quad \mbox{\rm and} \quad C = \frac{c_1 + 2 \lambda +  k_1}{c_1 + 2 \lambda -  k_1}.
$$
The value of $\mu$ is obtained from the second equation in system \eqref{1.19}  which yields
\begin{align*}
\varphi = \sigma \tilde{g}^{-1} (\mu	\tilde{h} f - h \tilde{f})
\end{align*}
with 
\begin{align*}
\mu	\tilde{h} f - h \tilde{f} &= (\mu - 1) \left(1 + \frac{c_1 + 2 \lambda +  k_1}{c_1 + 2 \lambda -  k_1}  e^{2 i k_1 \xi_1} \right) \\
& \quad + e^{i k_1 \xi_1} \left[ e^{\phi_1} \left( \mu \frac{c_1 + 2 \lambda +  k_1}{c_1 + 2 \lambda -  k_1}  - 1 \right) + e^{-\phi_1} \left( \mu - \frac{c_1 + 2 \lambda +  k_1}{c_1 + 2 \lambda -  k_1}  \right) \right].
\end{align*}
Due to the exact solution (\ref{1.35}), $\mu \tilde{h} f - h \tilde{f}$ must be divisible by $\tilde{g}$ which is true if and only if 
$$
\mu = 1 - \sigma \lambda^{-1}.
$$  
With this restriction on parameter $\mu \in \mathbb{R}$ of the linear system (\ref{1.19}), we obtain the explicit expression for the component $p$ of the eigenfunction:
\begin{equation}
\label{p-expression}
p = - \frac{\gamma_1 \lambda^{-1} }{1 + e^{i k_1 \xi_1 - \phi_1}} \left[ 1 + \frac{c_1 + 2 \lambda +  k_1}{c_1 + 2 \lambda -  k_1}  e^{i k_1 \xi_1 - \psi_1} \right].
\end{equation}
The third equation in system (\ref{1.19}) can be rewritten with the help of the first equation in (\ref{1.19}) in the form 
$$
i (D_t - \lambda D_x) \varphi \cdot f + i D_x h \cdot g = 0.
$$
Using (\ref{2.9}), (\ref{2.8}), (\ref{1.24}) with $\theta = 0$, and (\ref{1.35}) we have verified that this equation is satisfied. Thus, (\ref{q-plus-expression}), (\ref{q-minus-expression}), and (\ref{p-expression}) give the exact solution of (\ref{2.2}) for $q^- \neq 0$. The components $q^{\pm}$ are analytic in $\mathbb{C}_{\pm}$ and bounded 
as ${\rm Im}(x) \to \pm \infty$ for every $t \in \mathbb{R}$. 

According to Definition \ref{def-Lax}, we check the mean value of $q^-$ to obtain the additional bands $\cup_{j=0}^N [\lambda_j, \lambda_j + k_1]$ of the Lax spectrum. Since $q^-$ is analytic in $\mathbb{C}_-$, we use the geometric series to represent $q^-$ in the form
\begin{equation*}
q^- =  \left[ \frac{c_1 + 2 \lambda +  k_1}{c_1 + 2 \lambda -  k_1} + e^{-i k_1 \xi_1 - \phi_1} \right] \sum_{\ell=0}^{\infty} (-1)^{\ell} e^{-i \ell k_1 \xi_1 - \ell \phi_1},
\end{equation*}
from which it follows that the mean value of $q^-$ is zero at only one point given by 
$$
c_1 + 2 \lambda_0 + k_1 = 0.
$$
This yields only one additional band $[\lambda_0,\lambda_0+k_1]$ in the Lax spectrum given by  (\ref{Lax-spectrum}). 

Finally, we compare the location of $[\lambda_0,\lambda_0+k_1]$ relative to $[\sigma,\infty)$.
\begin{itemize}
	\item If $\sigma = +1$, then $c_1 \in (-2+k_1,-k_1)$ so that $\lambda_0 > 0$ and $\lambda_0 + k_1 < 1$ and $[\lambda_0,\lambda_0+k_1] \in (0,1)$ is isolated from $[1,\infty)$. 

\item If $\sigma = -1$, then either $c_1 \in (k_1+2,\infty)$ so that $\lambda_0 + k_1 < -1$ and $[\lambda_0,\lambda_0+k_1]$ is isolated from $[-1,\infty)$ or $c_1 \in (-\infty,-k_1)$ so that $\lambda_0 > -1$ and  $[\lambda_0,\lambda_0+k_1]$ is embedded into $[-1,\infty)$.
\end{itemize}
This completes the proof of proposition.
\end{proof}

\begin{remark}
	\label{rem-Lax-spectrum}
	It follows from (\ref{1.17}) and (\ref{p-expression}) that $p$ is analytic and bounded in $\mathbb{C}_+$. This was not a requirement on solutions of the linear system (\ref{2.2}). Nevertheless, since the spatial profile $u$ in the traveling periodic wave is analytic in $\mathbb{C}_+$, see Remark 
	\ref{rem-analytic}, the Lax spectrum of the linear operators $\mathcal{L}_u$ and $\mathcal{L}_u |_{L^2_+}$, see (\ref{linear-operator}) and (\ref{linear-operator-restr}), are identical to each other and $p$ is also analytic in $\mathbb{C}_+$.
\end{remark}

%\begin{remark}
%	The left edge $\lambda_0$ of the spectral band $[\lambda_0,\lambda_0+k_1]$ is uniquely determined from the zero mean value of $q^-$ when the second equation of system (\ref{2.2}), 
%	$$
%	q^+ - \mu q^- + \sigma \bar{u} p = 0,
%	$$
%	is considered in the space of $T$-periodic functions in $L^2_{\rm per}(0,T)$ with $T = \frac{2 \pi}{k_1}$. Indeed, the mean value of $q^+ \in L^2_{\rm per}(0,T) \cap L^2_+$ must be uniquely determined by the mean value of $-\sigma \bar{u} p$. On the other hand, the right edge $\lambda_0 + k_1$ can be obtained from the singularity of $p$ and $q^{\pm}$ given by (\ref{q-plus-expression}), (\ref{q-minus-expression}), and (\ref{p-expression}). 
%\end{remark}

\begin{remark}
	\label{rem-spectrum-soliton}
		The Lax spectrum of the algebraic soliton (\ref{2.12})--(\ref{2.13}) appears in the limit $k_1 \to 0$ of Proposition \ref{prop-2}. It consists of the spectral band $[\sigma,\infty)$ and a simple eigenvalue at $\lambda_0 = -\frac{c_1}{2}$. 
		\begin{itemize}
			\item If $\sigma = +1$, then $c_1 \in (-2,0)$ and $\lambda_0 \in (0,1)$ is isolated from the continuous spectrum $[1,\infty)$. 
			\item If $\sigma = -1$, then either $c_1 \in (2,\infty)$ and $\lambda_0 \in (-\infty,-1)$ is isolated from the continuous spectrum $[-1,\infty)$ or $c_1 \in (-\infty,0)$ and $\lambda_0 \in (0,\infty)$ is embedded into the continuous spectrum $[-1,\infty)$.
					\end{itemize}
\end{remark}

\subsection{Breathers on the traveling periodic wave}
\label{sec-breathers}

To obtain a solitary wave on the background of the traveling periodic wave (\ref{2.10})--(\ref{2.11}), we start with the $2$-periodic wave solution of the nonlocal model (\ref{INLS}) obtained in \cite{MatsunoPLA,MatsunoSAMP}. By using the representation (\ref{2.3}), we write the $2$-periodic wave solution in the form:
\begin{equation}
\label{2-periodic-wave}
\begin{cases}
 f &= 1 + e^{i k_1 \xi_1 - \phi_1-\frac{1}{2} A_{12}} + e^{i k_2 \xi_2 - \phi_2 -\frac{1}{2} A_{12}} + e^{i k_1 \xi_1 - \phi_1+ i k_2 \xi_2 - \phi_2}, \\
\tilde{f} & = 1 + e^{i k_1 \xi_1 + \phi_1-\frac{1}{2} A_{12}} + e^{i k_2 \xi_2 + \phi_2-\frac{1}{2} A_{12}} + e^{i k_1 \xi_1 + \phi_1+ i k_2 \xi_2 + \phi_2},	\\
g & = \gamma_{12} \left[ 1 + e^{i k_1 \xi_1 - \psi_1-\frac{1}{2} A_{12}} + e^{i k_2 \xi_2 - \psi_2-\frac{1}{2} A_{12}} + e^{i k_1 \xi_1 - \psi_1+ i k_2 \xi_2 - \psi_2} \right],	\\
\tilde{g} & = \gamma_{12}^{-1} \left[ 1 + e^{i k_1 \xi_1 + \psi_1-\frac{1}{2} A_{12}} + e^{i k_2 \xi_2 + \psi_2-\frac{1}{2} A_{12}} + e^{i k_1 \xi_1 + \psi_1+ i k_2 \xi_2 + \psi_2} \right], 
\end{cases}
\end{equation}
where $\xi_j = x - c_j t - x_{j}$ with arbitrary $x_j \in \mathbb{R}$, 
${\rm sgn}(\phi_j) = {\rm sgn}(k_j)$,  
\begin{align}
e^{2 \phi_j} = \frac{(c_j - k_j) (c_j + k_j + 2 \sigma)}{(c_j + k_j) (c_j - k_j + 2 \sigma)}, \quad e^{\psi_j} = \frac{c_j + k_j}{c_j - k_j} e^{\phi_j}, \quad j = 1,2,
\label{parameters-1}
\end{align}
\begin{align}
e^{-A_{12}} = \frac{(c_1 - c_2)^2 - (k_1 + k_2)^2}{(c_1 - c_2)^2 - (k_1 - k_2)^2},
\label{parameters-2}
\end{align}
and
\begin{align}
\gamma_{12} = e^{\frac12 (\psi_1 - \phi_1) + \frac12 (\psi_2 - \phi_2)}.
\label{parameters-3}
\end{align}
The parameters $k_{1,2}$ and $c_{1,2}$ must satisfy the same restrictions as in Proposition \ref{prop-1}:
\begin{itemize}
	\item If $\sigma = +1$, then $k_j \in (0,1)$ and $c_j \in (-2+k_j,-k_j)$, $j = 1,2$. 
	\item If $\sigma = -1$, then $k_j \in (0,\infty)$ and either $c_j \in (k_j+2,\infty)$ or $c_j \in (-\infty,-k_j)$, $j = 1,2$. 
\end{itemize}
In addition, the parameters must satisfy the restriction 
\begin{equation}
\label{par-restriction}
(c_1-c_2)^2 > (|k_1| + |k_2|)^2,
\end{equation}
which was proven in \cite[Lemma 1.1]{DK-91} for the BO equation. If the constraint (\ref{par-restriction}) is satisfied and ${\rm sgn}(k_j) = {\rm sgn}(\phi_j)$ for $j = 1,2$,  then the zeros of $f$ and $\tilde{f}$ are located in the lower and upper half-planes, respectively. This result of \cite{DK-91} holds for the nonlocal model (\ref{INLS}) because the functional representations of $f$ and $\tilde{f}$  in (\ref{2-periodic-wave})  is identical to that for the BO equation.

The following theorem gives the new breather solutions on the background of the traveling periodic wave. 

\begin{theorem}
	\label{prop-3}
	The nonlocal model (\ref{INLS}) admits breather solutions on the traveling periodic wave (\ref{2.10})--(\ref{2.11}). The solutions exist in the form (\ref{2.3}) with 
	\begin{equation}
	\label{expression-for-breathers}
	\begin{cases}
	f &= \left(1 - i \alpha_2 \xi_2 \right) \left( 1 + e ^{i k_1 \xi_1 - \phi_1} \right) + \alpha_2 \beta_{12}  \left(1 - e ^{i k_1 \xi_1 - \phi_1} \right),\\
	\tilde{f} &=   -\left(1 + i \alpha_2 \xi_2 \right)   \left( 1 + e ^{i k_1 \xi_1 + \phi_1} \right) + \alpha_2 \beta_{12} \left( 1 - e ^{i k_1 \xi_1 + \phi_1} \right), \\
	g  &= -\gamma_1  \left( 1 + \sigma c_2 + i \alpha_2 \xi_2 \right) \left( 1 + e ^{i k_1 \xi_1 - \psi_1} \right) + \gamma_1 \alpha_2 \beta_{12}  \left( 1 - e ^{i k_1 \xi_1 - \psi_1} \right), \\
	\tilde{g} &= \gamma_1^{-1} \left(1 + \sigma c_2 - i \alpha_2 \xi_2 \right) \left(1 + e ^{i k_1 \xi_1 + \psi_1} \right) + \gamma_1^{-1} \alpha_2 \beta_{12} \left( 1 - e ^{i k_1 \xi_1 + \psi_1} \right),
	\end{cases}
	\end{equation}
where
\begin{equation}
\label{beta-12}
\gamma_1 := \sqrt{\frac{c_1 + k_1}{c_1-k_1}} > 0, \;\; \alpha_2 := -\frac12 \sigma c_2 (c_2 + 2 \sigma) > 0, \;\; \beta_{12} := \frac{2 k_1}{(c_1 - c_2)^2 - k_1^2} > 0.
\end{equation}
Families of breather solutions are defined by the following intervals for admissible values of $c_2$:
\begin{itemize}
	\item If $\sigma = +1$, then $c_1 \in (-2+k_1,-k_1)$ and either 
	$c_2 \in (-2,c_1-k_1)$ or $c_2 \in (c_1+k_1,0)$. 
	\item If $\sigma = -1$, then 
	\begin{itemize}
		\item either $c_1 \in (k_1+2,\infty)$ with either $c_2 \in (c_1+k_1,\infty)$ or $c_2 \in (2,c_1-k_1)$ or $c_2 \in (-\infty,0)$ 
		\item or $c_1 \in (-\infty,-k_1)$ with either $c_2 \in (c_1+k_1,0)$ or $c_2 \in (-\infty,c_1-k_1)$ or $c_2 \in (2,\infty)$.
			\end{itemize}
\end{itemize}
\end{theorem}

\begin{proof}
We consider the $2$-periodic wave solution in the form (\ref{2-periodic-wave}) for $k_1 > 0$ and $k_2 > 0$ without loss of generality so that $\phi_1 > 0$ and $\phi_2 > 0$. By taking the long-wave limit $k_2\rightarrow 0$ in (\ref{parameters-1}) and (\ref{parameters-2}), we obtain the asymptotic expansions:
$$
\phi_2 = \frac{k_2}{\alpha_2} + \mathcal{O}(k_2^2), \quad \psi_2  = -\frac{(1 + \sigma c_2) k_2}{\alpha_2} + \mathcal{O}(k_2^2),
$$
and 
$$
e^{-\frac12 A_{12}} = 1 - \beta_{12} k_2 + \mathcal{O} ( k_2^2)
$$ 
with the corresponding expressions for $\alpha_2 > 0$ and $\beta_{12} > 0$ in (\ref{beta-12}). The expressions for $f$, $\tilde{f}$, $g$, and $\tilde{g}$ in (\ref{expression-for-breathers}) are obtained from (\ref{2-periodic-wave}) at the order of $\mathcal{O}(\phi_2)$ after the transformation $x_{2} \mapsto x_{2} + \pi / k_2$. The expression for $\gamma_1 > 0$ in (\ref{beta-12}) follows from the limit $k_2 \to 0$ of $\gamma_{12}$ in (\ref{parameters-3}).

Let us now analyze the constraints on $k_{1,2}$ and $c_{1,2}$ as well as the additional constraint (\ref{par-restriction}). As $k_2 \to 0$, it follows from (\ref{par-restriction}) that either $c_2 > c_1+k_1$ or $c_2 < c_1 - k_1$. 
\begin{itemize}
	\item If $\sigma = +1$, then we have $c_1 \in (-2+k_1,-k_1)$, and $c_2 \in (-2,0)$ with either $c_2 \in (-2,c_1-k_1)$ or $c_2 \in (c_1+k_1,0)$. 
	\item If $\sigma = -1$, then we have either $c_2 \in (2,\infty)$ or $c_2 \in (-\infty,0)$ with further intervals:
	\begin{itemize}
		\item If $c_1 \in (k_1+2,\infty)$, then either $c_2 \in (c_1+k_1,\infty) \subset (2,\infty)$ or $c_2 \in (2,c_1-k_1) \subset (2,\infty)$ or $c_2 \in (-\infty,0)$. 
		\item If $c_1 \in (-\infty,-k_1)$, then either $c_2 \in (c_1+k_1,0) \subset (-\infty,0)$ or $c_2 \in (-\infty,c_1-k_1) \subset (-\infty,0)$ or $c_2 \in (2,\infty)$.
			\end{itemize}
\end{itemize}

Finally, we show that the zeros of $f$ and $\tilde{f}$ given by (\ref{expression-for-breathers})  are located in the lower and upper half-planes, respectively, if $\alpha_2 > 0$ and $\beta_{12} > 0$. Without the loss of generality, we consider the zeros of $f$ given by 
$$
\left(1 - i \alpha_2 \xi_2 \right) \left( 1 + e ^{i k_1 \xi_1 - \phi_1} \right) + \alpha_2 \beta_{12}  \left(1 - e ^{i k_1 \xi_1 - \phi_1} \right) = 0,
$$
or equivalently, by 
$$
e ^{i k_1 \xi_1 - \phi_1} = \frac{\alpha_2 \beta_{12} + 1 - i \alpha_2 \xi_2}{\alpha_2 \beta_{12} - 1 + i \alpha_2 \xi_2}.
$$
Denote the root of this equation by $x = x_R + i x_J \in \mathbb{C}$ with $x_R = {\rm Re}(x)$ and $x_J = {\rm Im}(x)$. Then, $x_R$ and $x_J$ are obtained from 
$$
e ^{i k_1 \tilde{\xi}_1} e^{-k_1 x_J - \phi_1} = 
\frac{\alpha_2 \beta_{12} + 1 + \alpha_2 x_J - i \alpha_2 \tilde{\xi}_2}{\alpha_2 \beta_{12} - 1 - \alpha_2 x_J + i \alpha_2 \tilde{\xi}_2},
$$
where $\tilde{\xi}_j = x_R - c_j t - x_j$, $j= 1,2$ are real. Taking modulus in the equation yields
$$
e^{-k_1 x_J - \phi_1} = 
\sqrt{\frac{(\alpha_2 \beta_{12} + 1 + \alpha_2 x_J)^2 + \alpha_2^2 \tilde{\xi}^2_2}{(\alpha_2 \beta_{12} - 1 - \alpha_2 x_J)^2 + \alpha_2^2 \tilde{\xi}_2^2}}.
$$
If $x_J \geq 0$, this equation yields a contradiction since the left-hand side is less than $1$ and the right-hand side is larger than $1$, where we recall that $\alpha_2 > 0$, $\beta_{12}> 0$, and $\phi_1 > 0$. Hence, $x_J = {\rm Im}(x) < 0$ for every root of $f$. 
\end{proof}

\begin{remark}
	The last part of the proof of Theorem \ref{prop-3} is based on the proof in \cite[Section 2]{ChenPel24}, where it was overlooked that if $\xi_1$ has a nonzero imaginary part for the complex root of $f$, then $\xi_2$ also has a nonzero imaginary part. However, the same contradiction as in the proof of Theorem \ref{prop-3} can be obtained for the BO equation in \cite[Section 2]{ChenPel24}.
\end{remark}

Next we give several examples of the breather solutions.\\

\underline{If $\sigma = +1$,} the two intervals  $(c_1+k_1,0)$ and $(-2,c_1-k_1)$ for the wave speed $c_2$ in Theorem \ref{prop-3} are equivalent to 
\begin{equation}
\label{config-1}
		-\frac{c_2}{2} \in (0,\lambda_0) \quad \mbox{\rm and} \quad -\frac{c_2}{2} \in (\lambda_0 + k_1,1).
\end{equation}
The Lax spectrum of the corresponding breather solutions include an additional isolated eigenvalue $-\frac{c_2}{2} \in (0,1) \backslash [\lambda_0,\lambda_0+k_1]$ outside the Lax spectrum $\Sigma$ of the traveling periodic wave given in Proposition \ref{prop-2}, see also Remark \ref{rem-spectrum-soliton}. The two cases in (\ref{config-1}) are shown in Figure \ref{fig-Lax-1} for the particular parameters of Figures \ref{fig-2} and \ref{fig-3}.

\begin{figure}[htb!]
	\centering
	\includegraphics[width=7.5cm,height=4cm]{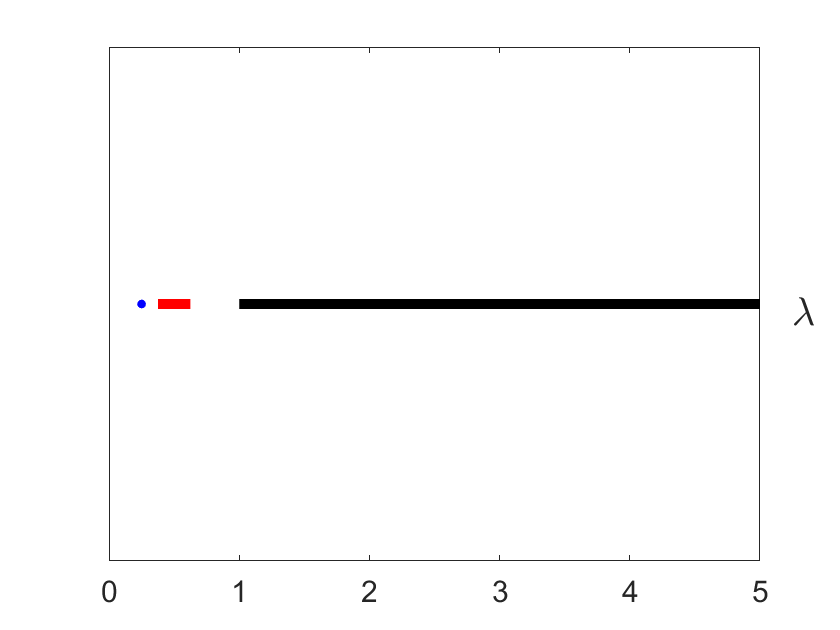}
	\includegraphics[width=7.5cm,height=4cm]{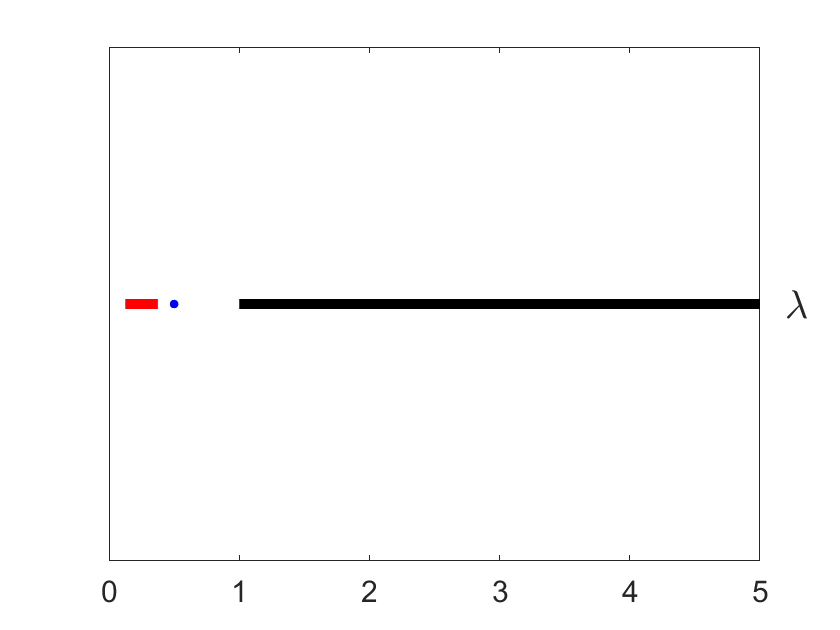}
	\caption{The Lax spectrum for the breather solutions of Figure \ref{fig-2} (left) and Figure \ref{fig-3} (right).}
	\label{fig-Lax-1}
\end{figure}

Figures \ref{fig-2} and \ref{fig-3} display the solution surfaces (side view on the left and top view on the right) for the breather solution of Theorem \ref{prop-3} for the two different choices in (\ref{existence-1}) and one choice in (\ref{config-1}). The solution surfaces are shown in the reference frame $x + t$ relative to the wave speed being equal to $-1$. In both cases, we can see that the breather solution represents the dark solitons over the traveling periodic wave.

Figure \ref{fig-2} is constructed for $k_1 = 0.25$, $c_1 = -1$, and 
$c_2 = -0.5 \in (c_1+k_1,0)$. The traveling periodic wave of Proposition \ref{prop-1} is stationary in the reference frame $x + t$, whereas the dark soliton in the breather solution propagates to the right direction relative to the periodic wave. Although the dark soliton impairs a phase shift, it is identically equal to the period of the traveling periodic wave. This property is in agreement with the fact that the limit $\xi_2 \to \pm \infty$ of the breather solutions of Theorem \ref{prop-3} yields the same traveling periodic wave of Proposition \ref{prop-1}. If $c_2 = -1.5 \in (-2,c_1-k_1)$, then the breather solution is very similar but the dark soliton moves slowly than the periodic wave (not shown).

Figure \ref{fig-3} is constructed for $k_1 = 0.25$, $c_1 = -0.5$, and 
$c_2 = -1 \in (-2,c_1-k_1)$. The dark soliton is stationary in the reference frame $x + t$, whereas the traveling periodic wave moves to the right direction relative to the dark soliton. We can see again that the phase shift of the breather is equal to the wave period. If $c_2 = -0.15 \in (c_1+k_1,0)$, the dark soliton moves faster than the periodic wave (not shown).\\

\begin{figure}[htb!]
	\centering
	\includegraphics[width=7.5cm,height=7cm]{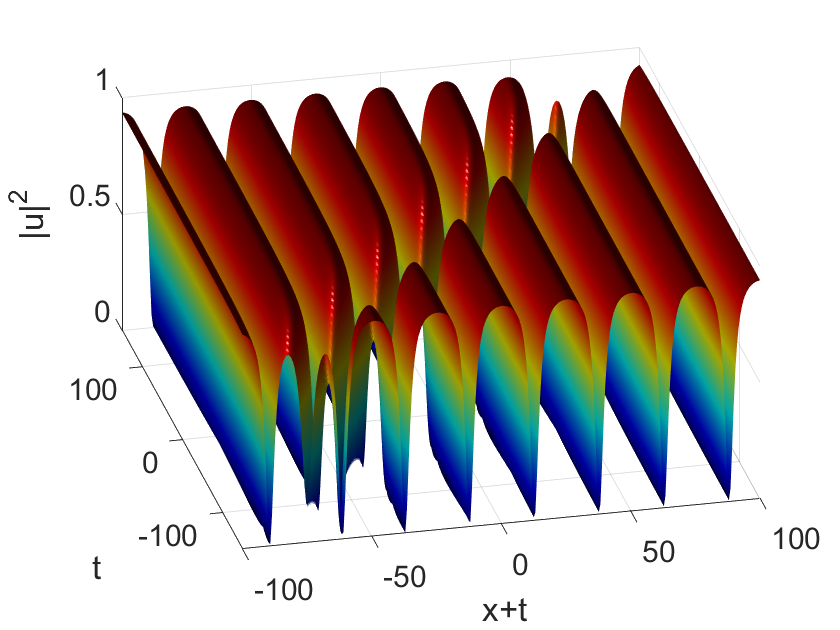}
	\includegraphics[width=7.5cm,height=7cm]{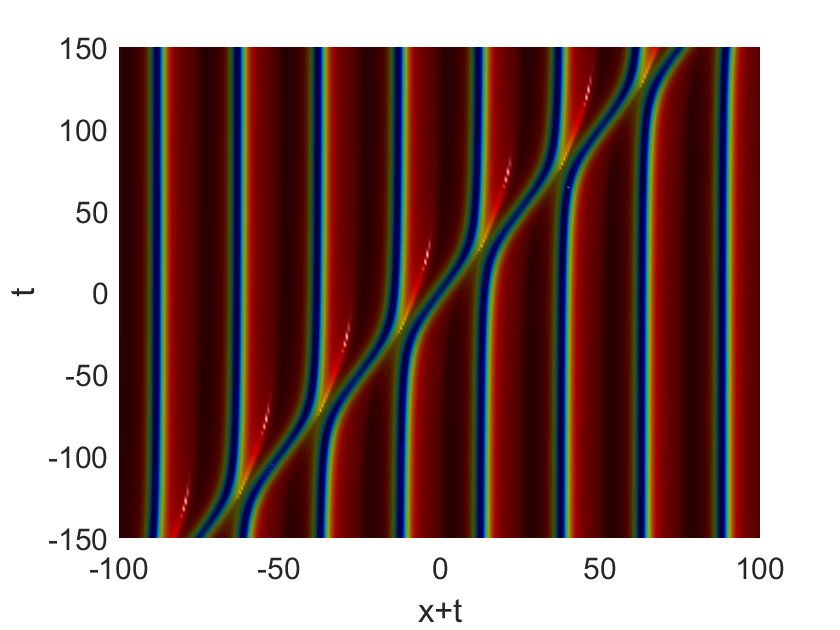}
	\caption{The solution surface of $|u|^2$ for the breather versus $(x+t,t)$ for $k_1 = 0.25$, $c_1 = -1$, and $c_2 = -0.5$.}
	\label{fig-2}
\end{figure}

\begin{figure}[htb!]
	\centering
	\includegraphics[width=7.5cm,height=7cm]{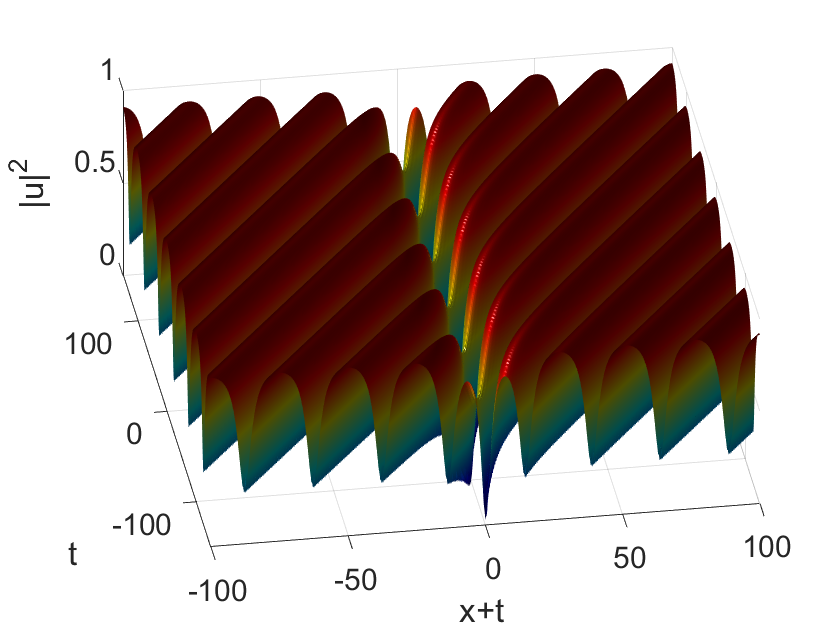}
	\includegraphics[width=7.5cm,height=7cm]{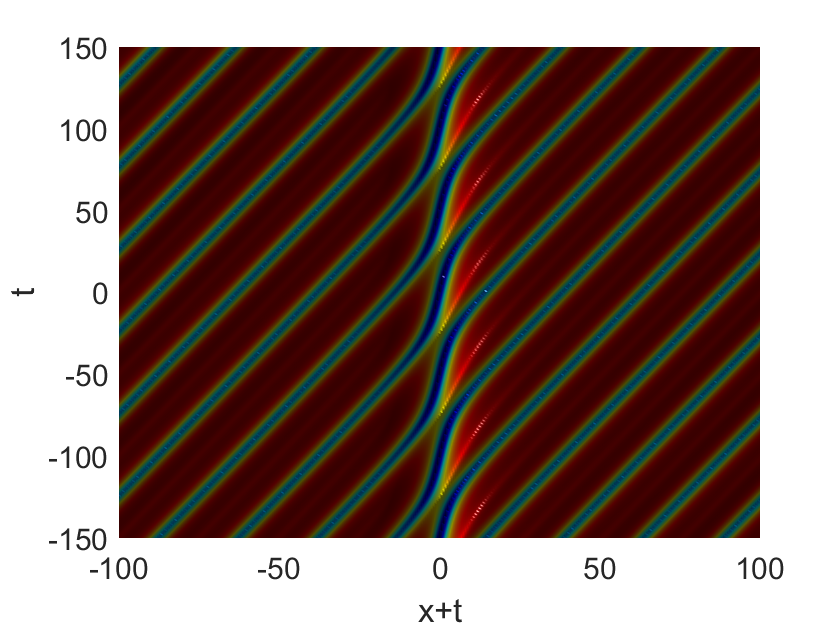}
	\caption{The solution surface of $|u|^2$ for the breather versus $(x+t,t)$ for $k_1 = 0.25$, $c_1 = -0.5$, and $c_2 = -1$.}
	\label{fig-3}
\end{figure}

\underline{If $\sigma = -1$ and $c_1 \in (k_1 + 2, \infty)$,} the three intervals $(c_1 + k_1, \infty)$, $(2, c_1 -  k_1)$, and $(-\infty,0)$ for the wave speed $c_2$ are equivalent respectively to
\begin{equation}
\label{config-2}
	- \frac{c_2}{2}  \in (- \infty, \lambda_0),\quad 
	- \frac{c_2}{2} \in ( \lambda_0 + k_1, -1), 
	 \quad {\rm and} \quad  -\frac{c_2}{2} \in (0,\infty).
\end{equation}
The Lax spectrum of the breather solutions includes an additional eigenvalue $- \frac{c_2}{2}$ relative to the Lax spectrum $\Sigma = [\lambda_0,\lambda_0+k_1] \cup [-1,\infty)$ of the traveling periodic wave. The eigenvalue $-\frac{c_2}{2}$ is isolated in the first two cases of (\ref{config-2}) and  embedded in the third case of (\ref{config-2}), see Proposition \ref{prop-2} and Remark \ref{rem-spectrum-soliton}. The first two cases in (\ref{config-2}) are shown in Figure \ref{fig-Lax-2} for the particular parameters of Figures \ref{fig-2-focus} and \ref{fig-3-focus}.

\begin{figure}[htb!]
	\centering
	\includegraphics[width=7.5cm,height=4cm]{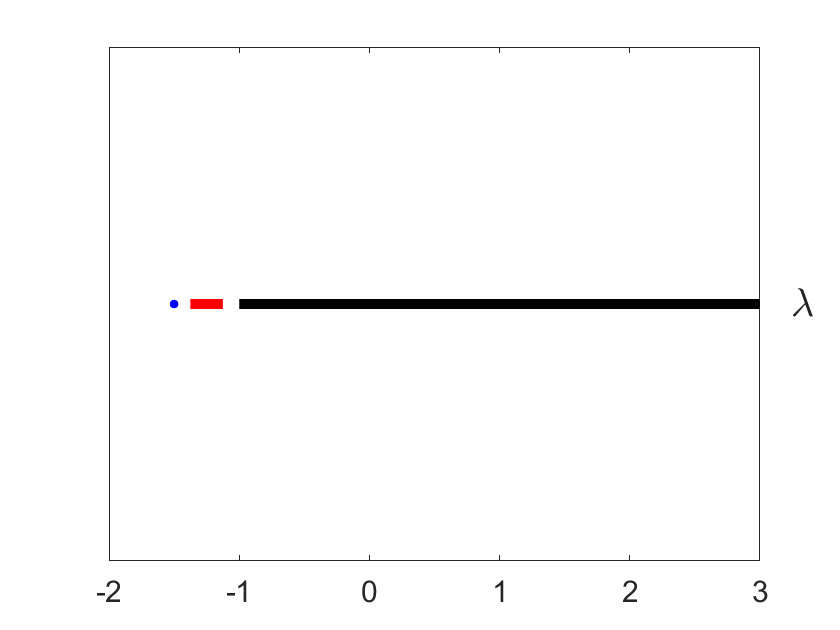}
	\includegraphics[width=7.5cm,height=4cm]{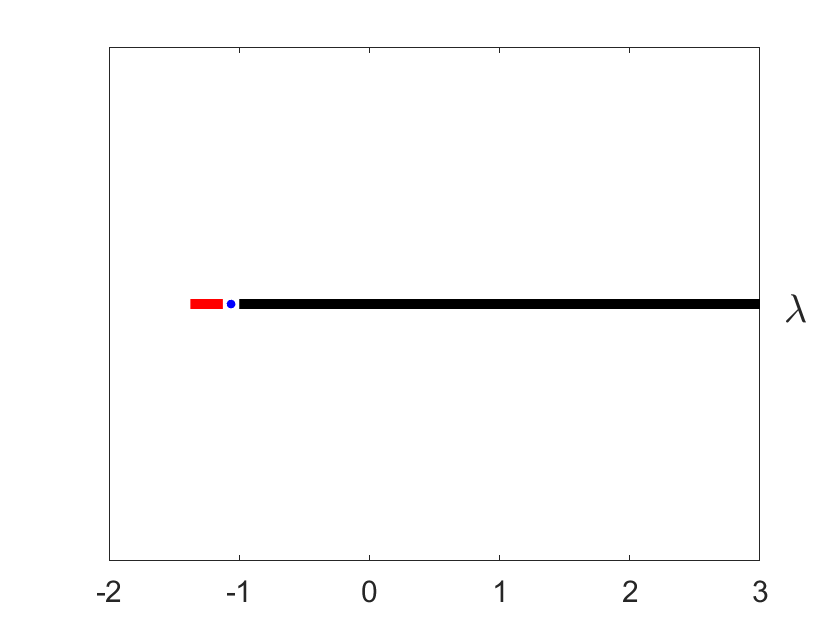}
	\caption{The Lax spectrum for the breather solutions of Figure \ref{fig-2-focus} (left) and Figure \ref{fig-3-focus} (right).}
	\label{fig-Lax-2}
\end{figure}

Figures \ref{fig-2-focus}, \ref{fig-3-focus},  and \ref{fig-4-focus} display the solution surfaces (side view on the left and top view on the right) for the breather solution of Theorem \ref{prop-3} with the three choices for $c_2$ in (\ref{config-2}). The solution surfaces are shown in the reference frame $x - c_1 t$, in which the periodic wave of Proposition \ref{prop-1} does not travel in time.

\begin{figure}[htb!]
	\centering
	\includegraphics[width=7.5cm,height=7cm]{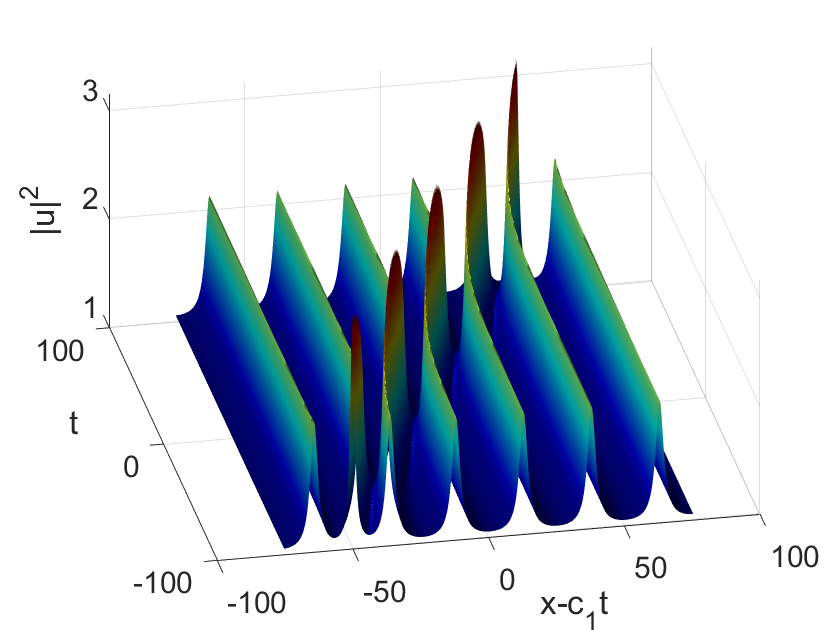}
	\includegraphics[width=7.5cm,height=7cm]{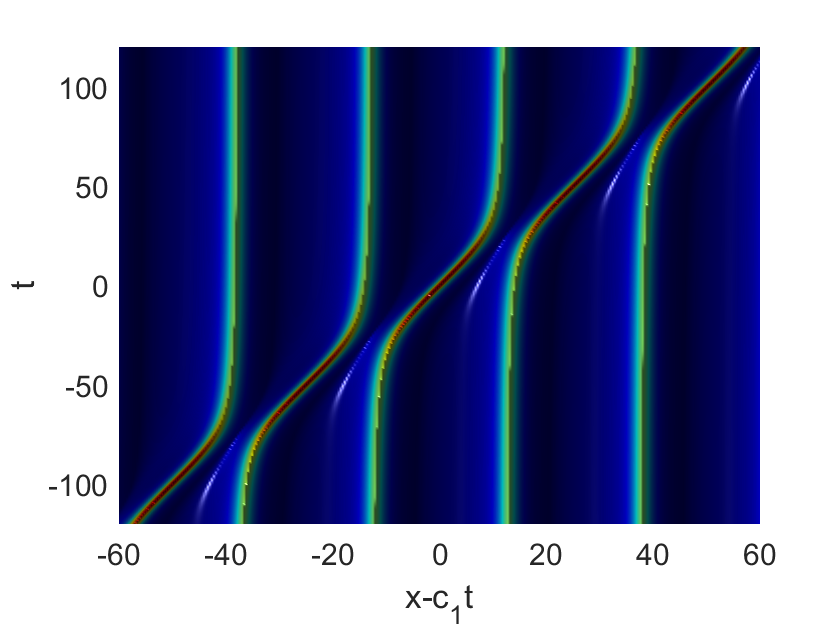}
	\caption{The solution surface of $|u|^2$ for the breather versus $(x-c_1 t,t)$ for $k_1 = 0.25$, $c_1 = 2+2k_1$, and $c_2 = c_1 + 2k_1$.}
	\label{fig-2-focus}
\end{figure}

\begin{figure}[htb!]
	\centering
	\includegraphics[width=7.5cm,height=7cm]{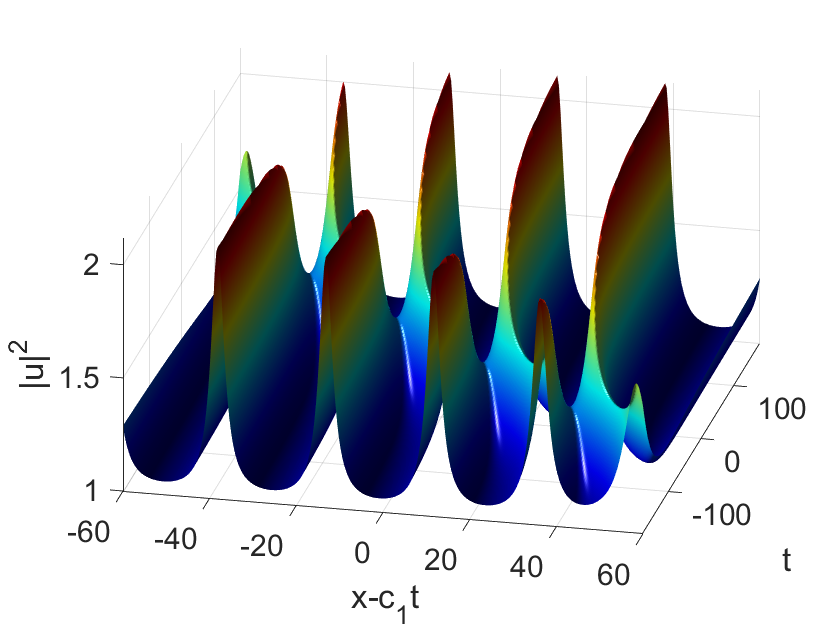}
	\includegraphics[width=7.5cm,height=7cm]{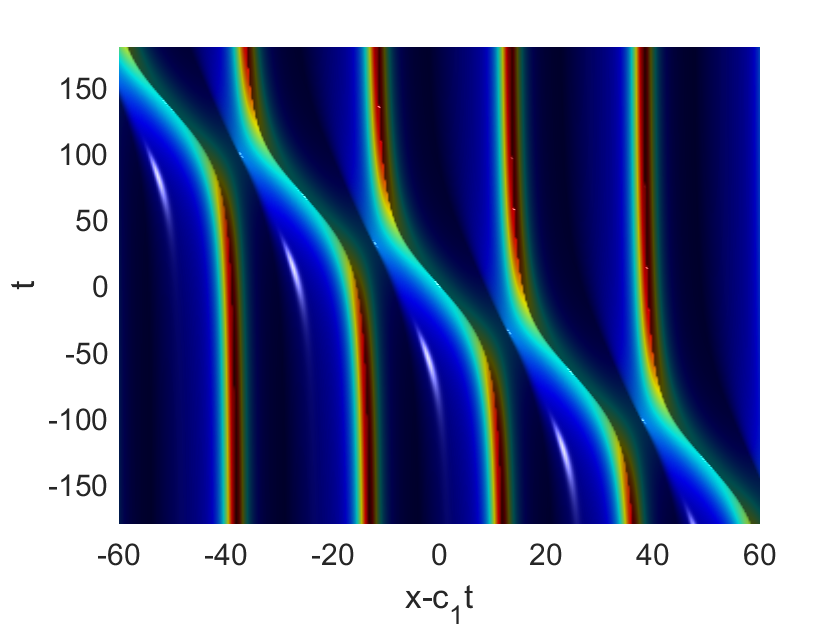}
	\caption{The solution surface of $|u|^2$ for the breather versus $(x-c_1 t,t)$ for $k_1 = 0.25$, $c_1 = 2+2k_1$, and $c_2 = c_1 - \frac{3}{2} k_1$.}
	\label{fig-3-focus}
\end{figure}

\begin{figure}[htb!]
	\centering
	\includegraphics[width=7.5cm,height=7cm]{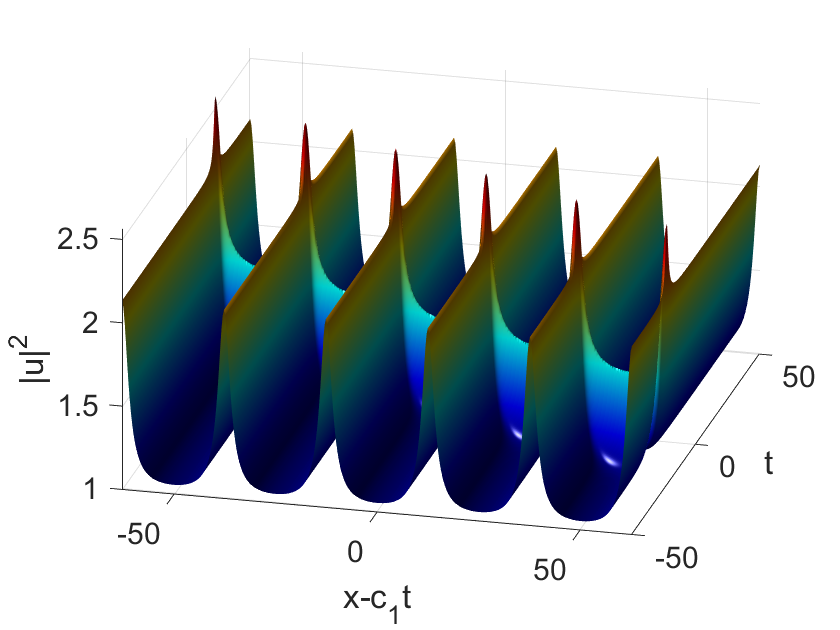}
	\includegraphics[width=7.5cm,height=7cm]{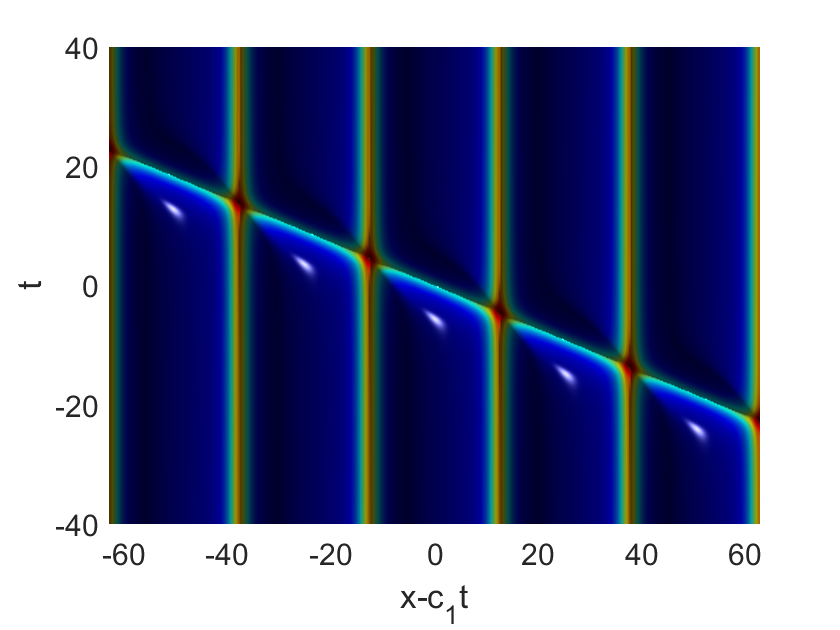}
	\caption{The solution surface of $|u|^2$ for the breather versus $(x-c_1 t,t)$ for $k_1 = 0.25$, $c_1 = 2+2k_1$, and $c_2 = -k_1$.}
	\label{fig-4-focus}
\end{figure}

Figure \ref{fig-2-focus} is constructed for $k_1 = 0.25$, $c_1 = 2 + 2k_1$, and 
$c_2 = c_1 + 2k_1 = 2 + 4k_1 \in (c_1+k_1,\infty)$. The breather solution has the bright soliton profile propagating to the right relative to the periodic wave. Figure \ref{fig-3-focus} is constructed for $k_1 = 0.25$, $c_1 = 2 + 2k_1$, and $c_2 = c_1 - \frac{3}{2} k_1 = 2 + \frac{1}{2} k_1 \in  (2,c_1-k_1)$. The breather solution has the dark soliton profile propagating to the left relative to the periodic wave. In both cases, we can clearly see that the phase shift of the breather is equal to the wave period of the traveling periodic wave. 

Figure \ref{fig-4-focus} is constructed for $k_1 = 0.25$, $c_1 = 2 + 2k_1$, and 
$c_2 = -k_1 \in (-\infty,0)$. The breather solution has the bright soliton profile propagatng to the left relative to the periodic wave. The breather doe not exhibit any phase shift.\\

\underline{If $\sigma = -1$ and $c_1 \in (- \infty, - k_1)$,} the periodic wave of Proposition \ref{prop-1} travels to the left symmetrically relative to the speed $1$, as follows from (\ref{existence-2}). The band $[\lambda_0,\lambda_0+k_1]$ of the Lax spectrum is now embedded into the continuous spectrum $[-1,\infty)$. The three intervals $(2,\infty)$, $(c_1 + k_1, 0)$, and $(-\infty, c_1 - k_1)$ for the wave speed $c_2$ are equivalent respectively to 
\begin{equation}
\label{config-3}
-\frac{c_2}{2} \in (-\infty,-1), \quad 
- \frac{c_2}{2} \in (0, \lambda_0), \quad {\rm and} \quad -\frac{c_2}{2} \in (\lambda_0 + k_1, \infty).
\end{equation}
The Lax spectrum of the breather solution includes and additional eigenvalue 
$-\frac{c_2}{2}$ relative to the Lax spectrum $\Sigma = [-1,\infty)$ of the traveling periodic wave. 
The eigenvalue is isolated in the first case of (\ref{config-3}) and embedded 
in the other two cases of (\ref{config-3}), see Proposition \ref{prop-2} 
and Remark \ref{rem-spectrum-soliton}. In both cases of the embedded eigenvalue $-\frac{c_2}{2}$, it is located outside the spectral band $[\lambda_0, \lambda_0 + k_1] \subset [-1, \infty)$. In spite of these differences between isolated and embedded eigenvalues, the breather solutions for all three cases in (\ref{config-3})  are very similar to the three cases in (\ref{config-2}) if they are plotted in the reference frame $x -c_1 t$ and $t$. The breather solution with $c_2 \in (2,\infty)$ has the bright soliton profile propagating to the right relative to the periodic wave (not shown). The breather solution with $c_2 \in (c_1+k_1,0)$ has the dark soliton profile propagating to the right relative to the periodic wave (not shown). The breather solution with  $c_2 \in (-\infty,c_1-k_1)$ has the bright soliton profile propagating to the left relative to the periodic wave (not shown). The breather solutions 
for (\ref{config-3}) can be obtained by reflections of the relevant solutions for (\ref{config-2}) due to the symmetry between cases $c_1 \in (-\infty,-k_1)$ and $c_1 \in (2+k_1,\infty)$ in (\ref{existence-2}).

\subsection{$N$-breathers on the traveling periodic wave}
\label{sec-n-breathers}

We obtain a general breather solution of the nonlocal model \eqref{INLS} for $N$ solitary waves propagating on the background of the traveling periodic wave (\ref{2.10})--(\ref{2.11}). To do so, we use the explicit formulas for the $(N+1)$-periodic wave solution obtained in \cite{MatsunoPLA}. The solution is written in the bilinear form (\ref{2.3}) with the following functions:
\begin{equation}
\label{2.30}
	f = \det F, \qquad 	\tilde{f} = \det \tilde{F}, \qquad g =  \gamma \det G, \qquad \tilde{g} = \gamma^{-1} \det \tilde{G},
\end{equation}
where matrices $F = (f_{jl})_{1 \leq j,l \leq N+1}$, $\tilde{F} = (\tilde{f}_{jl})_{1 \leq j,l \leq N+1}$, 
$G = (g_{jl})_{1 \leq j,l \leq N+1}$, and $\tilde{G} = (\tilde{g}_{jl})_{1 \leq j,l \leq N+1}$ are given by 
\begin{align*}
 f_{jl} &= \frac{1}{k_j} \exp{\left(i k_j \xi_j - \phi_j + \frac12 \sum_{s=1, s\neq j}^{N+1}
		A_{js}	\right) \delta_{jl}} + \frac{2}{c_j - c_l +k_j + k_l},  \\
\tilde{f}_{jl} &= \frac{1}{k_j} \exp{\left(i k_j \xi_j + \phi_j + \frac12 \sum_{s=1, s\neq j}^{N+1}
		A_{js}	\right) \delta_{jl}} + \frac{2}{c_j - c_l +k_j + k_l},  \\
g_{jl} &= \frac{1}{k_j} \exp{\left(i k_j \xi_j - \psi_j + \frac12 \sum_{s=1, s\neq j}^{N+1}
		A_{js}	\right) \delta_{jl}} + \frac{2}{c_j - c_l +k_j + k_l}, \\
\tilde{g}_{jl} &= \frac{1}{k_j} \exp{\left(i k_j \xi_j + \psi_j + \frac12 \sum_{s=1, s\neq j}^{N+1}
		A_{js}	\right) \delta_{jl}} + \frac{2}{c_j - c_l +k_j + k_l}, 
\end{align*}
with the parameters given by $\xi_j =x - c_j t - x_{j}$ with arbitrary $x_j \in  \mathbb{R}$, ${\rm sgn}(\phi_j) = {\rm sgn}(k_j)$,
\begin{align}
	e^{2 \phi_j}  = \frac{(c_j - k_j) (c_j + k_j +2 \sigma)}{(c_j + k_j) (c_j - k_j +2 \sigma)}, \quad e^{\psi_j} = \frac{c_j + k_j}{c_j - k_j} e^{\phi_j}, \quad 1 \leq j \leq N+1,
	\label{2.28} 
	\end{align}
	\begin{align}
e^{-A_{ij}}  = \frac{(c_i- c_j)^2 - (k_i + k_j)^2}{(c_i - c_j)^2 - (k_i - k_j)^2}, \quad 1 \leq i\neq j \leq N+1, 
\label{2.29}
\end{align}
and 
\begin{align}
\label{gamma-final}
	\gamma = \exp \left(\frac12 \sum_{j=1}^{N+1}(\psi_j - \phi_j)\right). 
\end{align}
We take the limit $k_j \to 0$ for $2 \leq j \leq N+1$. It follows from \eqref{2.28} and \eqref{2.29} that 
\begin{equation*} 
\phi_j = \frac{k_j}{\alpha_j} + \mathcal{O}(k_j^2), \qquad \psi_j  = -\frac{(1 + \sigma c_j) k_j}{\alpha_j} + \mathcal{O}(k_j^2),\qquad 2 \leq j \leq N+1,
\end{equation*}
and
\begin{align} 
\label{c2.40} 
e^{-\frac12 A_{1j}} &= 1 - \beta_{1j} k_j + \mathcal{O}(k_j^2), \quad 2 \leq j \leq N+1, \\
\label{c2.41}
e^{-\frac12 A_{ij}} &= 1 - \frac{2 k_i k_j}{(c_i - c_j)^2} + \mathcal{O}(k_i^2 k_j^2), \quad 2 \leq i, j \leq N+1,
\end{align}
where 
\begin{align*}
\alpha_j = -\frac{\sigma c_j (c_j + 2 \sigma)}{2}, \quad 
\beta_{1j} = \frac{2 k_1}{(c_1 - c_j)^2 - k_1^2}, \quad 2 \leq j \leq N+1.
\end{align*}
We multiply the first row of matrices $F$, $\tilde{F}$, $G$, and $\tilde{G}$ by 
$k_1$ and the remaining rows of these matrices by $\alpha_2,\cdots,\alpha_{N+1}$. The solution $u$ given by (\ref{2.3}) and (\ref{2.30}) is not affected by this transformation. Denoting 
the resulting matrices as $\hat{F}$, $\hat{\tilde{F}}$, $\hat{G}$, and $\hat{\tilde{F}}$, and taking the limit  as $k_j \rightarrow 0$ for $2 \leq j \leq N+1$, we obtain 
\begin{equation}
	\label{2.35}
	\hat{F} = \left( \begin{array}{cccc}
		1+ e^{ i k_1 \xi_1 - \phi_1} & \frac{2 k_1}{c_1 - c_2 + k_1}& \cdots & \frac{2 k_1}{c_1 - c_{N+1} + k_1}\\
		\frac{2\alpha_2}{c_2 - c_1 + k_1} & 	\hat{F}_{2 2}  & \cdots &  \frac{2\alpha_2}{c_2 - c_{N+1}}\\
		\vdots & \vdots & \ddots &\vdots\\
		\frac{2\alpha_{N+1}}{c_{N+1} -c_1 +k_1} &  \frac{2\alpha_{N+1}}{c_{N+1} - c_2}& \cdots & 	\hat{F}_{N+1 N+1}  \\
	\end{array}	\right),
\end{equation}

\begin{equation}
	\label{2.35-10}
	\hat{\tilde{F}} = \left( \begin{array}{cccc}
		1+ e^{ i k_1 \xi_1 + \phi_1} & \frac{2 k_1}{c_1 - c_2 + k_1}& \cdots & \frac{2 k_1}{c_1 - c_{N+1} + k_1}\\
		\frac{2\alpha_2}{c_2 - c_1 + k_1} & 	\hat{\tilde{F}}_{2 2}  & \cdots & \frac{2\alpha_2}{c_2 - c_{N+1}}\\
		\vdots & \vdots & \ddots &\vdots\\
		\frac{2\alpha_{N+1}}{c_{N+1} -c_1 +k_1} & \frac{2\alpha_{N+1}}{c_{N+1} - c_2}& \cdots & 	\hat{\tilde{F}}_{N+1 N+1}  \\
	\end{array}	\right),
\end{equation}

\begin{equation}
	\label{2.36}
	\hat{G} = \left( \begin{array}{cccc}
		1+ e^{ i k_1 \xi_1 - \psi_1} & \frac{2 k_1}{c_1 - c_2 + k_1}& \cdots & \frac{2 k_1}{c_1 - c_{N+1} + k_1}\\
		\frac{2\alpha_2}{c_2 - c_1 + k_1} &  \hat{G}_{2 2} & \cdots &  \frac{2\alpha_2}{c_2 - c_{N+1}}\\
		\vdots & \vdots & \ddots &\vdots\\
		\frac{2\alpha_{N+1}}{c_{N+1}-c_1  +k_1} &  \frac{2\alpha_{N+1}}{c_{N+1} - c_2}& \cdots & \hat{G}_{N+1 N+1}\\
	\end{array}	\right),
\end{equation}
and
\begin{equation}
	\label{2.36-10}
	 \hat{\tilde{G}} = \left( \begin{array}{cccc}
		1+ e^{ i k_1 \xi_1 + \psi_1} & \frac{2 k_1}{c_1 - c_2 + k_1}& \cdots & \frac{2 k_1}{c_1 - c_{N+1} + k_1}\\
		\frac{2\alpha_2}{c_2 - c_1 + k_1} &  \hat{\tilde{G}}_{2 2} & \cdots &  \frac{2\alpha_2}{c_2 - c_{N+1}}\\
		\vdots & \vdots & \ddots &\vdots\\
		\frac{2\alpha_{N+1}}{c_{N+1}-c_1  +k_1} & \frac{2\alpha_{N+1}}{c_{N+1} - c_2}& \cdots & \hat{\tilde{G}}_{N+1 N+1}\\
	\end{array}	\right),
\end{equation}
where
\begin{align*}
\hat{F}_{jj} &= - \alpha_j (i \xi_j + \beta_{1j}) + 1,\\
\hat{\tilde{F}}_{jj} &= - \alpha_j (i \xi_j + \beta_{1j}) - 1,\\
\hat{G}_{jj} &= - \alpha_j (i \xi_j + \beta_{1j}) - \sigma c_j - 1, \\
 \hat{\tilde{G}} _{jj} &= -\alpha_j (i \xi_j + \beta_{1j}) + \sigma c_j + 1,
\end{align*}
for $2 \leq j \leq N+1$. The general $N$-breather solution is obtained in the determinant form:
\begin{equation}
	\label{2.38}
u =  e^{\frac12 (\phi_1 - \psi_1)} \frac{ \det \hat{G}}{\det \hat{F}}, \qquad 	|u|^2 = 1 + \sigma  k_1 - i \sigma  \frac{\partial}{\partial x} \ln \frac{\det \hat{F}}{\det \bar{\hat{F}}},
\end{equation}
where $\hat{F}$ and $\hat{G}$ are given by \eqref{2.35} and \eqref{2.36}, with 
$$
\gamma_1 = e^{\frac12 (\phi_1 - \psi_1)} = \sqrt{\frac{c_1+k_1}{c_1-k_1}}
$$
obtained from (\ref{gamma-final}). Determinants $\hat{\tilde{F}}$ and  $\hat{\tilde{G}}$ obtained from \eqref{2.35-10} and \eqref{2.36-10} are related to the complex-cojugate versions of determinants $\hat{F}$ and $\hat{G}$. If $N=1$, then the solution (\ref{2.38}) recovers the breather solution of Theorem \ref{prop-3}.

\begin{remark}
	\label{rem-bounded}
	The general $N$-breather solution (\ref{2.38}) is bounded for every $(x,t) \in \R \times \R$. Indeed, it is obtained from the $(N+1)$-periodic wave solution for which zeros of $f$ and $\tilde{f}$ for every $k_1,k_2,\dots,k_{N+1} \neq 0$ are located in the lower and upper halves of the complex plane of $x$, respectively, due to \cite[Lemma 1.1]{DK-91}. As $k_2,\dots,k_{N+1} \to 0$, zeros of $f$ and $\tilde{f}$ can only approach the real line but cannot cross the real line. Since zeros of $f$ and $\tilde{f}$ on the real line coincide up to their multiplicities, their quotient is bounded for every $x \in \mathbb{R}$ at every $t \in \R$.
\end{remark}

\section{Traveling periodic waves and breathers on the zero background}
\label{sec-5}

Here we follow the structure of Section \ref{sec-4} but consider the traveling periodic waves at the zero background. The solutions of the nonlocal model (\ref{INLS}) on the zero background are only meaningful in the case of $\sigma = -1$, hence we take $\sigma = -1$ in what follows, see Remark \ref{rem-defoc}.

\subsection{Traveling periodic waves}

Assume that $f$ and $\tilde{f}$ have only zeros in the lower and upper half of the complex plane of $x$, respectively. Substitution
\begin{equation}
\label{3.3}
u = \frac{g}{f}, \quad \bar{u} = \frac{\tilde{g}}{\tilde{f}}, \quad |u|^2 = i \frac{\partial}{\partial x} \ln \frac{f}{\tilde{f}}.
\end{equation}
transforms the nonlocal model \eqref{INLS} with $\sigma = -1$ into the following system of bilinear equations:
\begin{equation}
\label{3.5}
\left\{ 
\begin{array}{l}
(i D_t + D_x^2) f \cdot g = 0, \\
(-i D_t + D_x^2) \tilde{f} \cdot \tilde{g} = 0, \\
i D_x f \cdot \tilde{f} - g \cdot \tilde{g} = 0.   
\end{array}  \right.
\end{equation}
The following proposition states the existence of the traveling periodic waves.

\begin{proposition}
	\label{prop-foc-1}
	The nonlocal model (\ref{INLS}) with $\sigma = -1$ admits the traveling periodic wave in the form
	\begin{equation}\label{3.14}
	u(x,t) = \frac{\gamma_1 e^{i k_1 \xi_1 - \phi_1}}{1 + e^{i k_1 \xi_1 - \phi_1}}, \quad 
	|u(x,t)|^2 = \frac{k_1 \sinh \phi_1}{\cos k_1 \xi_1 + \cosh \phi_1},
	\end{equation}	
	with $\xi_1 = x - c_1 t - x_1$, $c_1 = -k_1$, and $\gamma_1 = \sqrt{k_1 (e^{2 \phi_1}-1)}$, 
	where $k_1 > 0$, $\phi_1 > 0$, and $x_1 \in \R$ are arbitrary parameters. 
\end{proposition}

\begin{proof}
Let us consider  the following solution of the bilinear equations (\ref{3.5}):
	\begin{equation}
	\label{3.8-1} 
	\left\{ 
	\begin{array}{ll}
	f = 1 + e^{i k_1 \xi_1 - \phi_1}, &
	\tilde{f} = 1 + e^{i k_1 \xi_1 + \phi_1}, \\
	g = \gamma_1 e^{i k_1 \xi_1 - \phi_1}, & \tilde{g} = \gamma_1,
	\end{array}
	\right.
	\end{equation}
	where $\xi_1 = x - c_1 t - x_1$, with some parameters $k_1 \in \R$, $\phi_1 \in \R$, $c_1 \in \mathbb{R}$, $x_1 \in \mathbb{R}$, and $\gamma_1 \in \mathbb{R}$. The bilinear equations (\ref{3.5}) are satisfied if and only if 
	$$
	c_1 = -k_1, \quad \gamma_1^2 = k_1 (e^{2 \phi_1} - 1).
	$$
Without loss of generality, we can consider $k_1 > 0$, hence $\phi_1 > 0$ from requirement that $f$ and $\tilde{f}$ admit zeros only in the lower and upper half-planes, respectively. Since $\gamma_1^2 > 0$, then $\gamma_1 \in \R$ and $\bar{u}$ is complex conjugate of $u$.
\end{proof}

\begin{remark}
	\label{rem-defoc}
	The traveling periodic waves on the zero background are not admissible in the nonlocal model (\ref{INLS}) with $\sigma = +1$ because $\gamma_1^2 < 0$, which results in $\gamma_1 \in i \R$ and $|u(x,t)|^2 < 0$.
\end{remark}

\begin{remark}
		\label{rem-unity}
	Without the loss of generality, one can set $k_1 = 1$ due to the scaling transformation (\ref{NLS-scaling}). We will keep $k_1 > 0$ as a free parameter for clarity of notations.
\end{remark}

Figure \ref{fig-4} displays the spatial profile of the traveling periodic waves given by (\ref{3.14}) for $k_1 = 0.25$ and two choices of $\phi_1$. The wave with smaller $\phi_1$ has larger maxima and smaller minima for $|u|^2$ given by 
$$
\max_{\xi_1 \in \R} \frac{k_1 \sinh \phi_1}{\cos k_1 \xi_1 + \cosh \phi_1} = k_1 \coth\left(\frac{\phi_1}{2}\right)
$$
and
$$
\min_{\xi_1 \in \R} \frac{k_1 \sinh \phi_1}{\cos k_1 \xi_1 + \cosh \phi_1} = k_1 \tanh\left(\frac{\phi_1}{2}\right).
$$
The period of the traveling periodic wave is $\frac{2\pi}{k_1}$.

\begin{figure}[htb!]
	\centering
	\includegraphics[width=7.5cm,height=7cm]{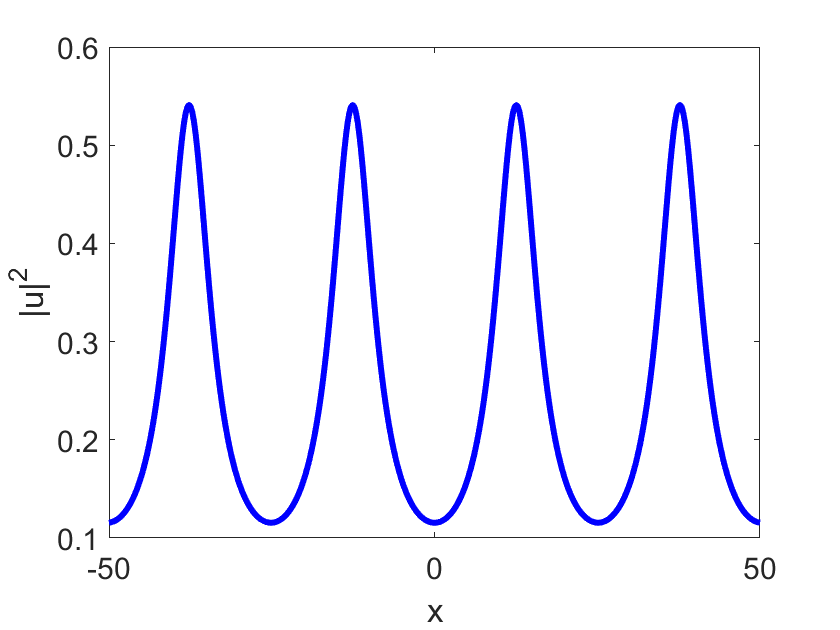}
	\includegraphics[width=7.5cm,height=7cm]{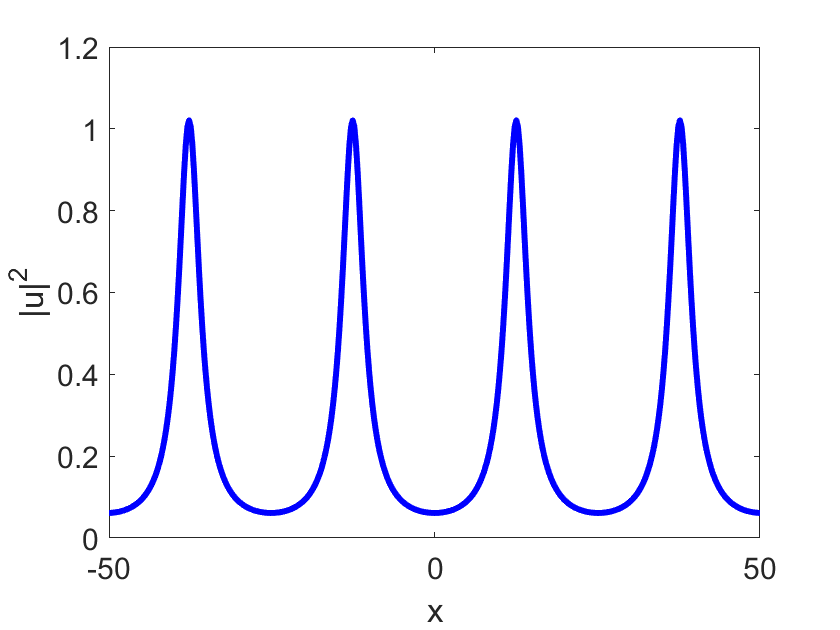}
	\caption{The profile of $|u|^2$ versus $x$ for $\sigma = -1$, $k_1 = 0.25$, and either $\phi_1 = 1$ (left) or $\phi_1 = 0.5$ (right).}
	\label{fig-4}
\end{figure}

The long-wave limit of the periodic wave (\ref{3.14}) appears as $k_1 \to 0$, where the periodic wave transforms into a solitary wave. As $k_1 \rightarrow 0$, the nontrivial limit exists in \eqref{3.14} if and only if $\phi_1 = \alpha_1 k_1 \rightarrow 0$ with arbitrary $\alpha_1 > 0$. We obtain from (\ref{3.14}) after the transformation $x_1 \mapsto x_1  + \pi/k_1$ in the limit $k_1 \to 0$ that 
\begin{equation} 
\label{3.15} 
u(x,t) = \frac{\sqrt{2\alpha_1}}{i (x - x_1) - \alpha_1}, \quad   
|u(x,t)|^2 = \frac{2\alpha_1}{\alpha_1^2 + (x-x_1)^2}.  
\end{equation}
The solution (\ref{3.15}) represents the bright soliton with the profile decaying to zero at infinity. The arbitrary parameter $\alpha_1 > 0$ can be normalized to unity due to the scaling transformation (\ref{NLS-scaling}).
Variational characterization of the soliton solution (\ref{3.15}) has been studied in \cite{GL-22}.

\subsection{Lax spectrum of the traveling periodic wave}

To obtain the exact solutions of the linear system (\ref{2.2}) with $\sigma = -1$, we use the representations (\ref{eig-bilinear}) and (\ref{3.3}). The system of bilinear equatons (\ref{1.19}) remains the same and we rewrite this system for $\sigma = -1$ as 
\begin{equation}
	\left\{
	\begin{array}{l}
		(i D_x + \lambda) \varphi \cdot f + g \cdot h  = 0, \\
		h \cdot \tilde{f} - \mu \tilde{h} \cdot f - \varphi \cdot \tilde{g} = 0, \\
		(i D_t + \lambda^2) \varphi \cdot f + (i D_x + \lambda) h \cdot g = 0, \\  
		(i D_t - 2 i \lambda D_x + D_x^2) h \cdot f = 0, \\   
		(i D_t - 2 i \lambda D_x + D_x^2) \tilde{h} \cdot \tilde{f} = 0.
	\end{array}  \right.
	\label{1.19-1}
\end{equation} 
The following proposition describes the Lax spectrum for the traveling periodic wave with the spatial profile (\ref{3.14}) based on the exact solutions of the system (\ref{1.19-1}).

\begin{proposition}
	\label{prop-foc-3}
	Let $u$ be the traveling periodic wave in Proposition \ref{prop-foc-1}. The Lax spectrum in Definition \ref{def-Lax} is located in 
\begin{equation}
\label{Lax-spectrum-foc-I}
\Sigma = [0,k_1] \cup [0,\infty),
\end{equation}
where $k_1 > 0$. Hence $[0,k_1]$ is embedded into $[0,\infty)$.	
\end{proposition}

\begin{proof}
	We proceed differently for $q^- \equiv 0$ and $q^- \neq 0$. \\
	
\underline{If $q^- \equiv 0$}, then $\tilde{h} \equiv 0$. The second solution of system (\ref{1.19-1}) implies that 
$$
h = \frac{\varphi \tilde{g}}{\tilde{f}}.
$$
Since $q^+ = \bar{u} p$ is analytic in $\mathbb{C}_+$, then $h$ is required to be analytic in $\mathbb{C}_+$. Since $\tilde{f}$ admits zeros in $\mathbb{C}_+$, then $\varphi$ must be divisible by $\tilde{f}$ so that 
\begin{equation} 
	\label{1.14-1}
	\varphi = m \tilde{f}, \qquad h =  m \tilde{g},
\end{equation}
with some $m = m(x,t)$ to be determined (required to be analytic in $\mathbb{C}_+$). 

From the first equation of 
system (\ref{1.19-1}) we find with the help of the third equaton of system (\ref{3.5}) that 
\begin{equation} 
	\label{m-eq-1}
	i m_x + \lambda m = 0.
\end{equation}
From the third equation of system (\ref{1.19-1}), we obtain with the help of (\ref{m-eq-1}) that
$$
\left( i m_t + \lambda^2 m \right) \tilde{f} \cdot f + i m \left( D_t \tilde{f}\cdot f + D_x \tilde{g} \cdot g  \right)= 0, 
$$
which together with \eqref{3.8-1} results in 
\begin{equation} 
	\label{m-eq-2}
	i m_t + \lambda^2 m = 0.
\end{equation}
From the fourth equation of system (\ref{1.19-1}) and  \eqref{3.8-1}, we obtain
$$
(i m_t - 2 i \lambda m_x + m_{xx}) f + m( -i f_t + 2 i \lambda f_x + f_{xx}) -2 m_x f_x = 0,
$$
which is satisfied due to  \eqref{3.8-1}, \eqref{m-eq-1}, and \eqref{m-eq-2}. Solving \eqref{m-eq-1} and \eqref{m-eq-2} yields
$$
m(x,t) = e^{i \lambda x + i \lambda^2 t},
$$
with the constant of integration normalized to unity.
By using (\ref{eig-bilinear}) and (\ref{1.14-1}), we obtain the exact expression for the components $p$ and $q^+$ of the eigenfunctions with $q^- \equiv 0$:
\begin{equation}
	\label{1.17-1}
	p = e^{i \lambda x + i \lambda^2 t} \frac{1 + e^{i k_1 \xi_1 + \phi_1}}{1 + e^{i k_1 \xi_1 - \phi_1}}, \qquad q^+ =  \frac{\gamma_1 e^{i \lambda x + i \lambda^2 t}}{1 + e^{i k_1 \xi_1 - \phi_1}}.
\end{equation}
The component $q^+$ is analytic in $\mathbb{C}_+$ and bounded as ${\rm Im}(x) \to +\infty$ for every $t \in \R$ if and only if 
$\lambda \geq 0$. Hence, $[0,\infty) \in \Sigma$ belongs to the Lax spectrum (\ref{Lax-spectrum-foc-I}). \\

\underline{If $q^- \neq 0$}, then we obtain solutions for $h$ and $\tilde{h}$ by using the last two equations of system \eqref{1.19-1}. Given $f$ in (\ref{3.8-1}), we separate the variables in the form
\begin{equation*}
	h = e^{i (\theta \xi_1 + \Omega t)} \left(1 + A e^{i k_1 \xi_1 - \phi_1}\right),
\end{equation*}
with some $\theta$, $\Omega$, and $A$ to be determined. The fourth equation in system \eqref{1.19-1} is satisfied if and only if 
\begin{equation*}
	\Omega = \theta (2 \lambda - \theta - k_1) \quad \mbox{\rm and} \quad A  = \frac{\lambda  - \theta}{ \lambda  - \theta -  k_1},
\end{equation*}
which yields the explicit solution 
\begin{equation} 
	\label{1.24-1}
	h = e^{i \theta (\xi_1 + (2 \lambda - \theta - k_1) t)} \left(1 + \frac{\lambda  - \theta}{ \lambda  - \theta -  k_1}  e^{i k_1 \xi_1 - \phi_1}\right).
\end{equation}
With similar computations from the fifth equation in system \eqref{1.19-1}, we obtain the explicit solution 
\begin{equation} \label{1.26-1}
	\tilde{h} = e^{i \theta (\xi_1 + (2 \lambda - \theta - k_1) t)} \left(1 + \frac{\lambda  - \theta}{ \lambda  - \theta -  k_1}  e^{i k_1 \xi_1 + \phi_1}\right),
\end{equation}
where parameter $\theta \in \mathbb{R}$ has to be the same due to the coupling between $h$ and $\tilde{h}$ in the second equation of system (\ref{1.19-1}). Now $q^+$ and $q^-$ are analytic and bounded in $\mathbb{C}_+$ and $\mathbb{C}_-$ respectively if and only if $\theta = 0$. This yields the unique representation of the components $q^+$ and $q^-$ in the form 
\begin{equation}
	\label{q-plus-expression-1}
	q^+ = \frac{1}{1 + e^{i k_1 \xi_1 - \phi_1}} \left[ 1 + \frac{\lambda }{ \lambda -  k_1}  e^{i k_1 \xi_1 - \phi_1} \right]
\end{equation}
and 
\begin{equation}
	\label{q-minus-expression-1}
	q^- = \frac{1}{1 + e^{i k_1 \xi_1 + \phi_1}} \left[ 1 + \frac{\lambda }{ \lambda -  k_1}  e^{i k_1 \xi_1 + \phi_1} \right]. 
\end{equation}
We note from (\ref{q-plus-expression-1}) and (\ref{q-minus-expression-1}) that $q^{\pm}$ are analytic and bounded in $\mathbb{C}_{\pm}$ according to Definition \ref{def-Lax} since $k_1 > 0$ and $\phi_1 > 0$.

It remains to find a bounded function $p$ from the first three equations of system (\ref{1.19-1}). Resorting to the second equation (\ref{1.19-1}) and (\ref{3.8-1}), we arrive at
\begin{equation}\label{2.27-0}
 \varphi = \frac{1}{\gamma_1} (h \tilde{f} - \mu \tilde{h} f) ,
\end{equation}
Substituting \eqref{3.8-1}, \eqref{2.27-0}, \eqref{1.24-1} and \eqref{1.26-1} with $\theta = 0$ into the first equation of (\ref{1.19-1}), we obtain
\begin{align*}
&\lambda(1 - \mu )	e^{3 i k_1 \xi_1 - \phi_1} + \lambda \left[1 + \frac{2 k_1 (\mu - 1)}{\lambda - k_1} + \frac{\lambda (1 - 2 \mu)}{\lambda - k_1}
+ e^{-2 \phi_1} \left(\frac{\lambda + \gamma_1}{\lambda - k_1} - \mu \right)\right] e^{2 i k_1 \xi_1} \\
& + \left[k_1 + \gamma_1 + 2 \lambda (1 - \mu) + e^{2 \phi_1} (\lambda (1 - \mu) - k_1) \right] e^{ i k_1 \xi_1 - \phi_1} + \lambda(1 - \mu ) = 0.
\end{align*}
This relation is satisfied if and only if $\mu =1 $. By using (\ref{2.27-0}), we obtain 
\begin{equation}\label{2.28-0}
\varphi = - \frac{\gamma_1}{\lambda - k_1}  e^{ i k_1 \xi_1 - \phi_1}.
\end{equation}
and 
\begin{equation}\label{2.29-0}
	p = - \frac{\gamma_1}{\lambda - k_1} \frac{e^{ i k_1 \xi_1 - \phi_1}}{1 + e^{ i k_1 \xi_1 - \phi_1}}  .
\end{equation}
Finally, we have confirmed that the third equation is satisfied by using \eqref{3.8-1}, \eqref{1.24-1} with $\theta = 0$ and \eqref{2.28-0}. 
The periodic function $q^-$ in (\ref{q-minus-expression-1}) 
has zero mean value at $\lambda = 0$. Hence the band $[0,k_1]$ belongs to the Lax spectrum (\ref{Lax-spectrum-foc-I}). 
\end{proof}

\begin{remark}
	It follows from (\ref{1.17-1}) and (\ref{2.29-0}) that $p$ is analytic and bounded in $\mathbb{C}_+$. We note again that although this was not a requirement on solutions of the linear system (\ref{2.2}), this property follows from the fact that the spatial profile $u$ in (\ref{3.14}) is analytic in $\mathbb{C}_+$, see also Remark \ref{rem-Lax-spectrum}.
\end{remark}

\begin{remark}
	The band $[0,k_1]$ of Proposition \ref{prop-foc-3} in (\ref{Lax-spectrum-foc-I}) can be formally obtained from the band $[\lambda_0,\lambda_0+k_1]$ of Proposition \ref{prop-2} in (\ref{Lax-spectrum}) since $\lambda_0 = -\frac{c_1+k_1}{2} = 0$ if $c_1 = -k_1$. The other band $[0,\infty)$ in (\ref{Lax-spectrum-foc-I}) is a shifted version of the band $[-1,\infty)$ in (\ref{Lax-spectrum}) due to the change from a nonzero to zero background.
\end{remark}

\begin{remark}
	The Lax spectrum of the algebraic soliton (\ref{3.15}) is found in the limit $k_1 \to 0$ of Proposition \ref{prop-foc-3}. It consists of the spectral band $[0,\infty)$ and a simple embedded eigenvalue at $0$, which is the end point of the continuous spectrum.
\end{remark}

\subsection{Breathers on the traveling periodic wave} 

To obtain a solitary wave on the background of the traveling periodic wave (\ref{3.14}), we first construct the $2$-periodic wave solution from the bilinear equations (\ref{3.5}). The following theorem gives the most general $2$-periodic solution, where we have verified  consistency of the assumptions on $f$ and $\tilde{f}$ to have only zeros in the lower and upper halves of the complex plane of $x$, respectively. 

\begin{theorem}
	\label{prop-tech}
	The nonlocal model (\ref{INLS}) with $\sigma = -1$ admits two families of the $2$-periodic wave solution expressed by the representation \eqref{3.3} with 
\begin{equation} 
\label{2.39-0}
\begin{cases}
		f &= 1 + e^{i k_1 \xi_1 - \phi_1 - \frac12 A_{12}} + e^{i k_2 \xi_2 - \frac12 A_{12}} + e^{i k_1 \xi_1 - \phi_1+ i k_2 \xi_2},\\
		\tilde{f} & = 1 + e^{i k_1 \xi_1 + \phi_1 - \frac12 A_{12}} + e^{i k_2 \xi_2 - \frac12 A_{12}} + e^{i k_1 \xi_1 + \phi_1+ i k_2 \xi_2},\\
		g & = \gamma_{12} \left[ e^{i k_1 \xi_1 - \phi_1 - \frac12 A_{12}} + \alpha_2 e^{i k_1 \xi_1 - \phi_1+ i k_2 \xi_2} \right],\\
		\tilde{g} &= \gamma_{12} \left[e^{i k_2 \xi_2 - \frac12 A_{12}} + \alpha_2 \right],
	\end{cases}
\end{equation}
where $\xi_j = x - c_j t - x_j$ with arbitrary $x_j \in \mathbb{R}$, $k_1 > 0$, $\phi_1 > 0$, $c_1 = -k_1$, $k_2 \neq 0$, 
\begin{align*}
\alpha_2 = \frac{c_2  - k_2}{c_2 + k_2},
\qquad e^{- A_{12}} = \frac{(c_1 -c_2)^2 - (k_1+k_2)^2}{(c_1-c_2)^2 - (k_1-k_2)^2} = \frac{(c_2  - k_2)(c_2 + 2 k_1 + k_2)}{(c_2  + k_2)(c_2 + 2 k_1 - k_2)},
\end{align*}
and
$$
\gamma_{12}^2 =  k_1 (e^{2 \phi_1} - 1) \frac{c_2 + k_2}{c_2 - k_2}.
$$
The two families are defined by the two intervals for the speed $c_2$: 
either $c_2 \in (|k_2|,\infty)$ or $c_2 \in (-\infty,-|k_2|-2k_1)$.
\end{theorem}

\begin{proof}
We have verified validity of the explicit expression (\ref{2.39-0}) by substituting them into the bilinear equation (\ref{3.5}). The complex conjugate symmetry between $u$ and $\bar{u}$ in \eqref{3.3} holds if and only if $\gamma_{12}^2 > 0$. This corresponds to $\alpha_2 > 0$ since $k_1 > 0$ and $\phi_1 > 0$. Therefore, either $c_2 > |k_2|$ or $c_2 < -|k_2|$. 

Next we prove that zeros of $f$ and $\tilde{f}$ are located in the lower and upper halves of the complex plane of $x$, respectively, if 
$$
(c_1 - c_2)^2 > (k_1 + |k_2|)^2,
$$
for $k_1 > 0$ and $k_2 \neq 0$. Since $c_1 = -k_1$, this constraint is rewritten as 
$$
(c_2 - |k_2|)(c_2 + |k_2| + 2 k_1) > 0,
$$
hence it provides further restrictions on the speed $c_2$: either $ c_2 \in (|k_2|,\infty)$ or $c_2 \in (-\infty,-|k_2|-2k_1)$.

In order to consider zeros of $f$, we note that $e^{-A_{12}} \to 1$ if $k_1 \to 0$ or $k_2 \to 0$ which yields the factorization formula:
$$
f \to (1 + e^{i k_1 \xi_1 - \phi_1}) (1 + e^{i k_2 \xi_2}) \quad \mbox{\rm as } \;\; e^{-A_{12}} \to 1.
$$
Zeros of the factorization for $f$ correspond to either ${\rm Im}  \xi_1 = - \phi_1/ k_1 < 0$ or ${\rm Im}  \xi_2 = 0$. The first set is already in the lower half-plane. To prove that the second set moves to the lower half-plane for small $k_1 > 0$ and $k_2 \neq 0$, we compute the perturbation terms beyond the factorization formula. The first-order Taylor expansion of $f$ in variable $k_1$ is as follows:
$$
f =  (1 + e^{ - \phi_1}) (1 + e^{i k_2 \xi_2}) + k_1 \left[i \xi_1 e^{- \phi_1} (1 + e^{i k_2 \xi_2}) - \frac{2 k_2}{c_2^2 - k_2^2}(e^{- \phi_1}  + e^{i k_2 \xi_2})\right] + \mathcal{O}(k_1^2).
$$
Let $\xi_0 = (2 n +1) \pi / k_2$ be a simple zero of $(1 + e^{ - \phi_1}) (1 + e^{i k_2 \xi_2}) = 0$ for $n \in \mathbb{Z}$ so that 
\begin{align*}
\partial_{\xi_2} f &= i k_2 (1 + e^{- \phi_1}) e^{i k_2 \xi_0} + \mathcal{O}(k_1) \\
&= - i k_2 (1 + e^{ - \phi_1}) + \mathcal{O}(k_1) \neq 0.
\end{align*}
By using the implicit function theorem, we obtain the root of $f = 0$ for small $k_1 > 0$ as 
$$
\xi_2 = \xi_0 - \frac{2 i k_1 (1 - e^{- \phi_1})}{(c_2^2 - k_2^2) (1 + e^{- \phi_1})} + \mathcal{O}(k_1^2) \quad \mbox{\rm as} \;\; k_1 \to 0.
$$
Since $(c_2-k_2) = \alpha_2 (c_2+k_2)$ and $\alpha_2 > 0$, then $c_2^2 - k_2^2 >0$. Since $k_1 > 0$ and $\phi_1>0$, we conclude that ${\rm Im}(\xi_2) < 0$ for every $n \in \mathbb{Z}$, hence all roots of $f$ are the lower half-plane for small $k_1 > 0$ and fixed $k_2 \neq 0$.

Let us now show that the zeros of $f$ do not cross the real axis for arbitrary $k_1 > 0$ and $k_2 \neq 0$. If $f = 0$ and $\xi_1,\xi_2 \in \mathbb{R}$, then we have the system 
\begin{align*}
\left\{ \begin{array}{r} 
1 + e^{-\phi_1 - \frac12 A_{12}} \cos k_1 \xi_1 + e^{- \frac12 A_{12}} \cos k_2 \xi_2 + e^{-\phi_1} \cos(k_1 \xi_1 + k_2 \xi_2) = 0, \\
e^{-\phi_1 - \frac12 A_{12}} \sin k_1 \xi_1 + e^{- \frac12 A_{12}} \sin k_2 \xi_2 + e^{-\phi_1} \sin (k_1 \xi_1 + k_2 \xi_2) = 0. \end{array} \right.
\end{align*}
By using 
$$
k_1 \xi_1 = \frac{k_1 \xi_1 + k_2 \xi_2}{2} + \frac{k_1 \xi_1 - k_2 \xi_2}{2}, \qquad k_2 \xi_2 =  \frac{k_1 \xi_1 + k_2 \xi_2}{2} - \frac{k_1 \xi_1 - k_2 \xi_2}{2}
$$
and trigonometric identities for addition/subtraction formulas, we rewrite the system in the equivalent form:
\begin{align*}
&	\left[ \cos \frac12 (k_1 \xi_1 + k_2 \xi_2) + e^{- \frac12 A_{12}} \cos \frac12 (k_1 \xi_1 - k_2 \xi_2)  \right] (e^{- \phi_1} + 1) \cos \frac12 (k_1 \xi_1 + k_2 \xi_2) \\ & - \left[ \sin \frac12 (k_1 \xi_1 + k_2 \xi_2) + e^{- \frac12 A_{12}} \sin \frac12 (k_1 \xi_1 - k_2 \xi_2)  \right] (e^{- \phi_1} - 1) \sin \frac12 (k_1 \xi_1 + k_2 \xi_2) = 0, \\ & \left[ \sin \frac12 (k_1 \xi_1 + k_2 \xi_2) + e^{- \frac12 A_{12}} \sin \frac12 (k_1 \xi_1 - k_2 \xi_2)  \right] (e^{- \phi_1} - 1)  \cos \frac12 (k_1 \xi_1 + k_2 \xi_2) \\ & + \left[ \cos \frac12 (k_1 \xi_1 + k_2 \xi_2) + e^{- \frac12 A_{12}} \cos \frac12 (k_1 \xi_1 - k_2 \xi_2)  \right] (e^{- \phi_1} + 1)  \sin \frac12 (k_1 \xi_1 + k_2 \xi_2) = 0.
\end{align*}
Since $\phi_1 > 0$, eliminating $\cos \frac12 (k_1 \xi_1 + k_2 \xi_2)$ and $\sin \frac12 (k_1 \xi_1 + k_2 \xi_2)$ yields the equivalent system of equations:
\begin{align*}
\left\{  \begin{array}{l}
\left[ \cos \frac12 (k_1 \xi_1 + k_2 \xi_2) + e^{- \frac12 A_{12}} \cos \frac12 (k_1 \xi_1 - k_2 \xi_2)  \right]  = 0, \vspace{2mm} \\
\left[ \sin \frac12 (k_1 \xi_1 + k_2 \xi_2) + e^{- \frac12 A_{12}} \sin \frac12 (k_1 \xi_1 - k_2 \xi_2)  \right] = 0,\end{array} \right.
\end{align*}
which can be rewritten further as 
\begin{equation}\label{3.20-10}
\left\{  \begin{array}{l}
\left(1 +  e^{- \frac12 A_{12}}  \right) \cos \frac12 k_1 \xi_1 \cos \frac12 k_2 \xi_2 - \left(1 - e^{- \frac12 A_{12}} \right) \sin \frac12 k_1 \xi_1 \sin \frac12 k_2 \xi_2 = 0,  \vspace{2mm}  \\
\left(1 +  e^{- \frac12 A_{12}} \right)\sin \frac12 k_1 \xi_1 \cos \frac12 k_2 \xi_2 + \left(1 - e^{- \frac12 A_{12}} \right) \cos \frac12  k_1 \xi_1 \sin \frac12 k_2 \xi_2 = 0. 
\end{array} \right.
\end{equation}
The determinant of coefficients for $\cos \frac{1}{2} k_2 \xi_2$ and $\sin \frac{1}{2} k_2 \xi_2$ in \eqref{3.20-10} is equal to $\xi_1$-independent quantity
$$ 
1 - e^{-  A_{12}} = \frac{4 k_1 k_2}{(c_2 + k_2)(c_2 - k_2 + 2 k_1)},
$$
which is nonzero and bounded if $k_1 > 0$, $k_2 \neq 0$ and either $c_2 \in (|k_2|,\infty)$ or $c_2 \in (-\infty,-|k_2|-2k_1)$. Hence, no zeros of $f$ cross the real line for arbitrary $k_1 > 0$ and $k_2 \neq 0$, and, since they are in the lower half-plane for small $k_1 > 0$ and fixed $k_2\neq 0$, they remain in the lower half-plane for every $k_1 > 0$ and $k_2 \neq 0$.
\end{proof}

\begin{remark}
The $N$-periodic solution of the nonlocal model (\ref{INLS}) with $\sigma = -1$ was recently obtained in \cite[Theorem 1]{MatsunoSAMP}. However, the parameter restriction in equation (14) of \cite{MatsunoSAMP} only recovers the family of solutions in our Theorem \ref{prop-tech} with $c_2 \in (|k_2|,\infty)$ and miss the family of solutions with $c_2 \in (-\infty,-|k_2|-2k_1)$. Although the location of zeros of $f$ in the lower half-plane of $x$ was proven in \cite[Proposition 2]{MatsunoSAMP} for parameter restriction in equation (14), we gave an alternative proof of Theorem \ref{prop-tech} which works for both families of solutions. 
\end{remark}

Each family of the $2$-periodic waves in Theorem \ref{prop-tech} generates
only one family of breathers on the background of the traveling periodic wave of Proposition \ref{prop-foc-1} in the limit $k_2 \to 0$, according to the following corollary.

\begin{corollary}
	\label{prop-foc-5}
	The nonlocal model (\ref{INLS}) with $\sigma = -1$ admits two families of breather solutions on the traveling periodic wave (\ref{3.14}). The solutions exist in the form (\ref{3.3}) with 
	\begin{equation} 
	\label{2.40-0}
	\left\{ \begin{array}{l}
	f = \beta_{12} (1 - e^{i k_1 \xi_1 - \phi_1}) - i \xi_2 (1 + e^{i k_1 \xi_1 - \phi_1}), \\
	\tilde{f} = \beta_{12} (1 - e^{i k_1 \xi_1 + \phi_1}) - i \xi_2 (1 + e^{i k_1 \xi_1 + \phi_1}), \\
	g = \gamma_1 e^{i k_1 \xi_1 - \phi_1} (\chi_{12}  - i \xi_2), \\
	\tilde{g} = \gamma_1 (- \chi_{12} -  i \xi_2), 
	\end{array}\right.
	\end{equation}
	where $\xi_j = x - c_j t - x_j$ with arbitrary $x_j \in \mathbb{R}$, $k_1 > 0$, $\phi_1 > 0$, $c_1 = -k_1$, 
	$$
	\gamma_1^2 = k_1 (e^{2 \phi_1} - 1) > 0, \quad 
	\beta_{12} = \frac{2k_1}{c_2 (c_2 + 2k_1)} > 0, \quad \mbox{\rm and} \quad  \chi_{12} = \frac{2(c_2 + k_1)}{c_2(c_2 + 2k_1)}.
	$$
The two families are defined by the two intervals: either $c_2 \in (0,\infty)$ or $c_2\in (-\infty,-2k_1)$.
\end{corollary}

\begin{proof}
	We obtain in the long-wave limit $k_2 \to 0$ that 
	$$
\alpha_2 = 1 - \frac{2k_2}{c_2} + \mathcal{O}(k_2^2), \quad	e^{-A_{12}} = 1 - \frac{4 k_1 k_2}{c_2 (c_2 + 2k_1)} + \mathcal{O}(k_2^2).
	$$
Hence $e^{-\frac{1}{2} A_{12}} = 1 - \beta_{12}k_2 + \mathcal{O}(k_2^2)$. The expressions (\ref{2.40-0}) are obtained from 
(\ref{2.39-0}) at the order of $\mathcal{O}(k_2)$ after the transformation $x_2 \mapsto x_2 + \pi/k_2$. The two families of $2$-periodic waves in Theorem \ref{prop-tech} give the two families of breathers with either $c_2 \in (0,\infty)$ or $c_2\in (-\infty,-2k_1)$ for which $\beta_{12}> 0$.

Next we prove that the zeros of $f$ and $\tilde{f}$ are located in the lower and upper half-planes, respectively, for either $c_2 \in (0,\infty)$ or $c_2\in (-\infty,-2k_1)$. Setting $f = 0$ in \eqref{2.40-0} yields
$$
e^{i k_1 \xi_1 - \phi_1} = \frac{\beta_{12} - i \xi_2}{\beta_{12} + i \xi_2}.
$$
Similarly to the proof of Theorem \ref{prop-3}, we denote the root of this equation by $x = x_R + i x_J \in \mathbb{C}$ with $x_R = {\rm Re}(x)$ and $x_J = {\rm Im}(x)$. Then, $x_R$ and $x_J$ are obtained from 
$$
e ^{i k_1 \tilde{\xi}_1} e^{-k_1 x_J - \phi_1} = 
\frac{\beta_{12} + x_J - i \tilde{\xi}_2}{\beta_{12} - x_J + i \tilde{\xi}_2},
$$
where $\tilde{\xi}_j = x_R - c_j t - x_j$, $j= 1,2$ are real. Taking modulus in the equation yields
$$
e^{-k_1 x_J - \phi_1} = 
\sqrt{\frac{(\beta_{12}  + x_J)^2 + \tilde{\xi}^2_2}{(\beta_{12} - x_J)^2 + \tilde{\xi}_2^2}}.
$$
If $x_J \geq 0$, this equation yields a contradiction since the left-hand side is less than $1$ for $k_1, \phi_1 > 0$ and the right-hand side is larger than $1$ for $\beta_{12} > 0$. Hence, $x_J = {\rm Im}(x) < 0$ for every root of $f$. 
\end{proof}

\begin{remark}
	The intervals $(0,\infty)$ and $(-\infty,-2k_1)$ for parameter $c_2$ are equivalent to 
	$$
	-\frac{c_2}{2} \in (-\infty,0) \quad \mbox{\rm and} \quad -\frac{c_2}{2} \in (k_1,\infty).
	$$
The Lax spectrum of the corresponding breather solutions include an additional eigenvalue $-\frac{c_2}{2}$ outside the band $[0,k_1]$ in the Lax spectrum $\Sigma$ in Proposition \ref{prop-foc-3}. However, there are differences between these two cases. For $c_2 \in (0,\infty)$, the new eigenvalue is isolated from $\Sigma$ and for $c_2 \in (-\infty,-2k_1)$, the new eigenvalue is embedded into $\Sigma \backslash [0,k_1]$. 
\end{remark}

\begin{figure}[htb!]
	\centering
	\includegraphics[width=7.5cm,height=7cm]{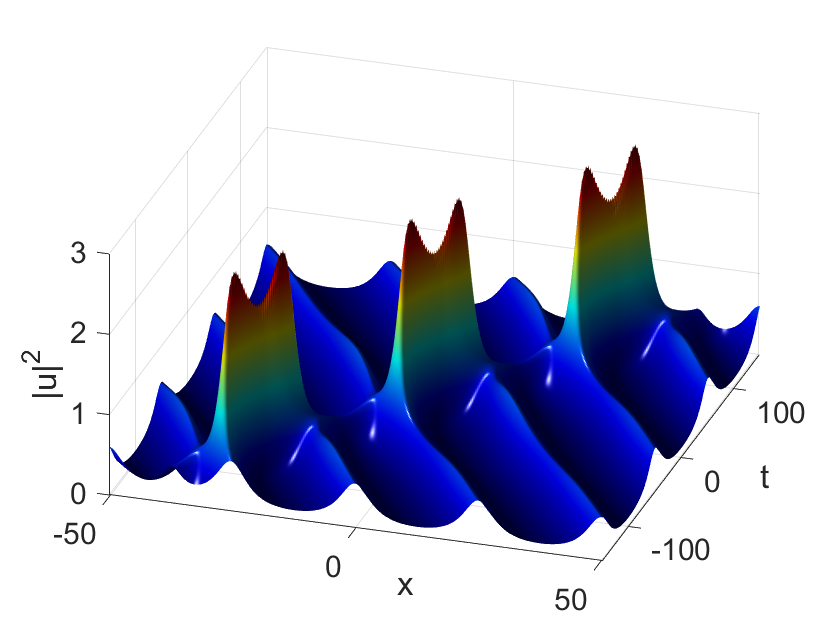}
	\includegraphics[width=7.5cm,height=7cm]{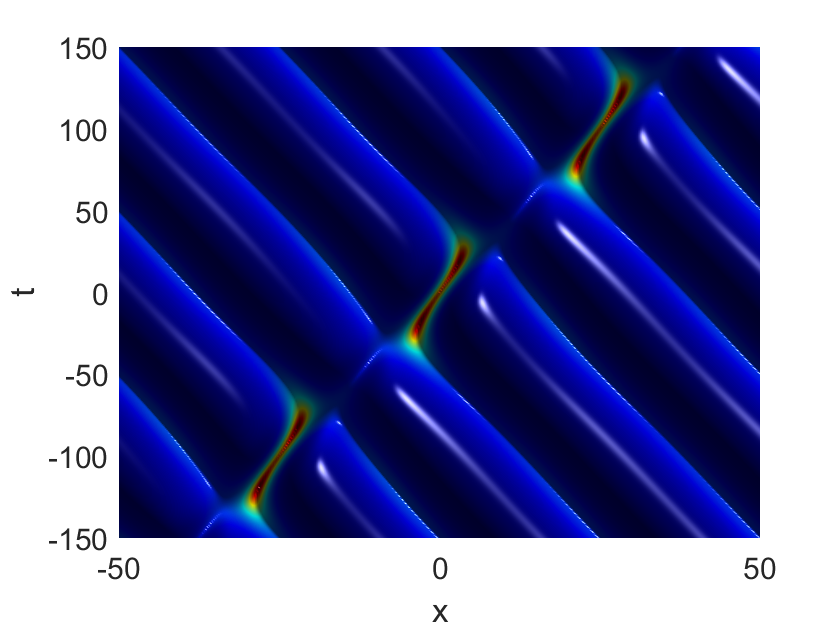}
	\caption{The solution surface of $|u|^2$ for the breather versus $(x,t)$ for $k_1 = 0.25$, $\phi_1 = 1$, $c_1 = -k_1$, and $c_2 = k_1$.}
	\label{fig-5}
\end{figure}

\begin{figure}[htb!]
	\centering
	\includegraphics[width=7.5cm,height=7cm]{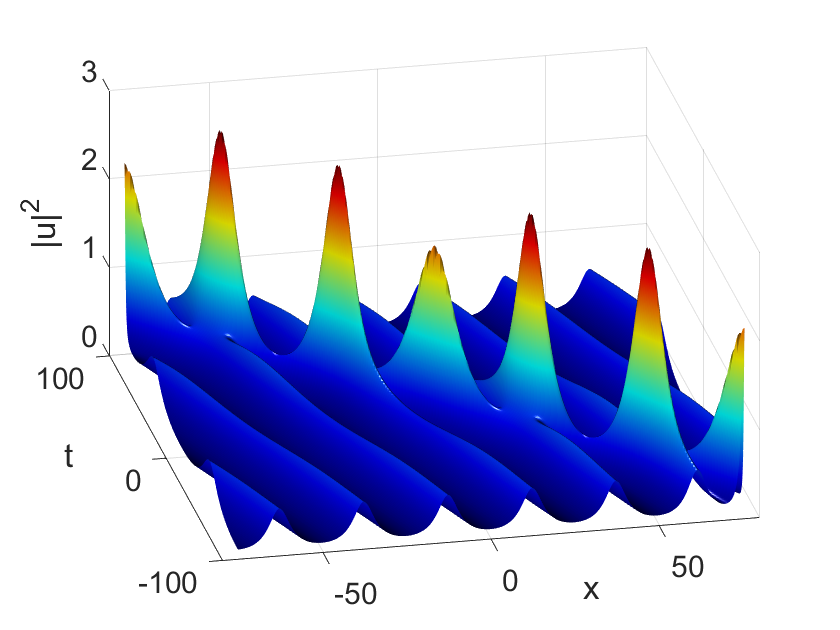}
	\includegraphics[width=7.5cm,height=7cm]{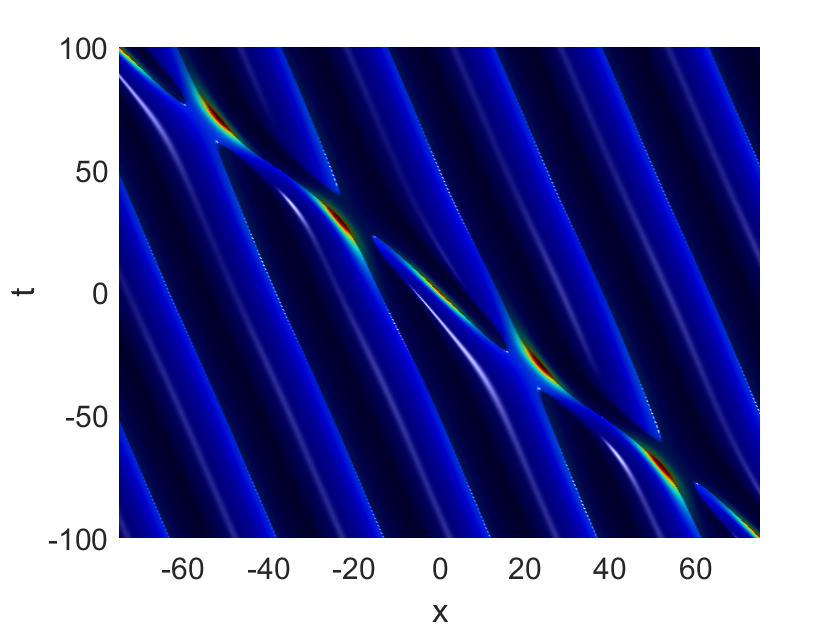}
	\caption{The solution surface of $|u|^2$ for the breather versus $(x,t)$ for $k_1 = 0.25$, $\phi_1 = 1$, $c_1 = -k_1$, and $c_2 = -3 k_1$.}
	\label{fig-6}
\end{figure}

Figures \ref{fig-5} and \ref{fig-6} display the solution surfaces (side view on the left and top view on the right) for the breather solution of Corollary \ref{prop-foc-5} with two choices for $c_2$ versus the original variables $(x,t)$. For $k_1 = 0.25$, $\phi_1 = 1$, and $c_2 = k_1 \in (0,\infty)$ on Figure \ref{fig-5},  the breather solution has the bright soliton profile propagating to the right of the periodic wave traveling to the left. For $k_1 = 0.25$, $\phi_1 = 1$, and $c_2 = -3k_1 \in (-\infty,-2k_1)$ on Figure \ref{fig-6}, the breather solution has the bright soliton profile propagating to the left faster than the periodic wave. In both caes, no phase shift appears after the bright soliton passes the traveling periodic wave.

\subsection{$N$-breathers on the traveling periodic wave} 

We obtain a general breather solution of the nonlocal model  \eqref{INLS} with $\sigma = -1$ for $N$ solitary waves propagating on the background of the traveling periodic wave (\ref{3.14}). To do so, we use the explicit formulas for the $(N+1)$-periodic wave solution obtained in \cite{MatsunoSAMP} but extend the choices for $\{ k_2,\dots,k_{N+1}\}$ and $\{ c_2,\dots,c_{N+1}\}$ in a more general setting as in Theorem \ref{prop-tech}. 

The $(N+1)$-periodic wave solution is given by 
\begin{equation}
\label{N-wave-solution}
u = \frac{g}{f}, \qquad 
|u|^2 = - \sum_{j=1}^{N+1} k_j + i \frac{\partial}{\partial x} \ln \frac{f}{\bar{f}}
\end{equation}	
	 with 
\begin{equation*}
f = \det(F), \qquad  g = \gamma \det(G), 
\end{equation*}
where $F = (f_{jl})_{1 \leq j,l \leq N+1}$ and $G = (g_{jl})_{1 \leq j,l \leq N+1}$ are given by 
\begin{align*}
f_{jl} = \frac{1}{k_j} \exp{\left(i k_j \xi_j - \phi_j + \frac12 \sum\limits_{k=1, k\neq j}^{N+1} A_{jk}\right)} \delta_{jl} + \frac{2}{c_j - c_l + k_j +k_l}
\end{align*}
and
\begin{align*}
g_{jl} = \left\{ \begin{array}{lll} \frac{1}{k_1} \exp\left( i k_1 \xi_1 - \phi_1 + \frac12 \sum\limits_{k=1, k\neq j}^{N+1} A_{jk} \right), \quad & 1 \leq j \leq N+1, \;\; & l = 1, \\
f_{jl}, \quad & 1 \leq j \leq N+1, \;\; & 2 \leq l \leq N+1, 
\end{array} \right.
\end{align*}
with $\xi_j = x - c_j t - x_{j}$ for arbitrary $x_j \in \R$, with $k_1 > 0$, $\phi_1 > 0$, $c_1 = -k_1$ and $k_j \neq 0$, $\phi_j = 0$ for $2 \leq j \leq N+1$, and with 
\begin{align*}
\gamma^2 = k_1 (e^{2 \phi_1} - 1) \prod_{j=2}^{N+1} \frac{c_j - k_j}{c_j + k_j}
\end{align*}
and 
\begin{align*}
e^{-A_{ij}}  = \frac{(c_i - c_j)^2 - (k_i + k_j)^2}{(c_i- c_j)^2 - (k_i - k_j)^2}, \quad 1 \leq i\neq j \leq N+1.
\end{align*}
The admissible intervals for $c_j$, $2 \leq j \leq N+1$ are given by 
$c_j \in (|k_j|, \infty)$ or $c_j \in (- \infty, - |k_j| - 2k_1)$ for $2 \leq j \leq N+1$. In addition, we have 
$$
(c_i - c_j)^2 > (|k_i| + |k_j|)^2, \qquad  2 \leq i, j \leq N+1.
$$

We take the limit $k_j \to 0$ for $2 \leq j\leq N+1$ and use the expansions \eqref{c2.40} and \eqref{c2.41} again with 
$$
\beta_{1j} = \frac{2k_1}{(c_1-c_j)^2 - k_1^2} =  \frac{2k_1}{c_j (c_j + 2 k_1)} > 0.
$$
Multiplying the first row of $F$ and the first column of $G$ by $k_1$ and taking the limit, we denote the resulting matrices as $\hat{F}$ and $\hat{G}$ and obtain them in the form:
\begin{equation}
\label{3.25}
\hat{F} = \left( \begin{array}{cccc}
1+ e^{ i k_1 \xi_1 - \phi_1} & - \frac{2 k_1}{c_2}& \cdots & - \frac{2 k_1}{ c_{N+1} }\\
\frac{2 }{c_2 + 2 k_1} & - i \xi_2 - \beta_{12} & \cdots & \frac{2}{c_2 - c_{N+1}}\\
\vdots & \vdots & \ddots &\vdots\\
\frac{2 }{c_{N+1} +2 k_1} & \frac{2 }{c_{N+1} - c_2}& \cdots & - i \xi_{N+1} - \beta_{1 N+1} \\
\end{array}	\right)
\end{equation}
and
\begin{equation}
\label{3.26}
\hat{G} = \left( \begin{array}{cccc}
e^{i k_1 \xi_1 - \phi_1}  & - \frac{2}{c_2} & \cdots & - \frac{2}{ c_{N+1} } \\
e^{i k_1 \xi_1 - \phi_1} & - i \xi_2 - \beta_{12} & \cdots & \frac{2}{c_2 - c_{N+1}}\\
\vdots & \vdots & \ddots &\vdots\\
e^{i k_1 \xi_1 - \phi_1} & \frac{2 }{c_{N+1} - c_2}& \cdots & - i \xi_{N+1} - \beta_{1 N+1} \\
\end{array}	\right).
\end{equation}
The general $N$-breather solution is obtained from (\ref{N-wave-solution}) in the closed determinant form
\begin{equation}
\label{3.27}
u = \gamma_1   \frac{\det \hat{G}}{\det \hat{F}}, \qquad |u|^2 = -k_1 + i \partial_x \ln \frac{\det \hat{F}}{\det \bar{\hat{F}}},
\end{equation}
where $\gamma_1^2 = k_1 (e^{2\phi_1}-1)$, whereas $\hat{F}$ and $\hat{G}$ are given by \eqref{3.25} and \eqref{3.26}. The wave speeds satisfy either $c_j \in (0, \infty)$ or $c_j \in (- \infty,-2k_1)$ for $2 \leq j \leq N+1$. In particular, if $N=1$, then the solution \eqref{3.27} with  \eqref{3.25} and \eqref{3.26} recovers the breather solution  \eqref{2.40-0} in Corollary \ref{prop-foc-5}.

\begin{remark}
	The general $N$-breather solution (\ref{3.27}) is bounded for every $(x,t) \in \R \times \R$ due to the same reason as in Remark \ref{rem-bounded} since it is obtained from the $(N+1)$-periodic wave solution for which zeros of $f$ and $\tilde{f}$ for every $k_1,k_2,\dots,k_{N+1} \neq 0$ are located in the lower and upper halves of the complex plane of $x$, respectively, due to \cite[Proposition 2]{MatsunoSAMP}. 
\end{remark}

\section{Conclusion}
\label{sec-concl}

We have studied the nonlocal derivative NLS equation, a new emerging model for deep fluids to describe modulations of wave packets and the continuum limit in the dynamics of particles. For the defocusing version of this nonlocal model, we proved the linear stability of the nonzero constant background for decaying and periodic perturbations and the nonlinear stability for periodic perturbations. 
For the focusing version, we proved the linear stability under a non-resonance condition on the initial data and the nonlinear stability for sufficiently small periods.

We have systematically studied the traveling periodic waves, their Lax spectrum, and the existence of breathers propagating on the background of the traveling periodic waves. In the defocusing case, there is only one family of traveling periodic waves on the nonzero constant background and we have shown existence of exactly two families of single breathers, both have the dark (depression) profiles. In the focusing case, there are two families of traveling periodic waves on the nonzero constant background, each admits three families of single breathers, two of which have the bright (elevation) profiles and one has the dark (depression) profile. Also in the focusing case, there is only one family of traveling periodic waves on the zero background and two families of single breathers, both have the bright (elevation) profiles. We have related the existence of breathers with the bands of the Lax spectrum for the traveling periodic waves. Surprisingly, breathers associated with either isolated or embedded eigenvalues in the Lax spectrum feature the same dynamics. Multi-breather solutions are obtained in the closed determinant form. 

We conclude by formulating further questions related to this study which can be considered in near future.

\begin{enumerate}
	\item Can the nonlinear stability or instability of the nonzero constant background be proven in the focusing case? What is the long-term dynamics of the linearly (algebraically) growing periodic perturbations at the resonance due to the nonlinear effects?\\
	
	\item Do there exist any exact solutions describing nonlinear instability of the traveling periodic waves in the focusing case? On comparison with the focusing cubic NLS equation with the modulational instability of the traveling periodic waves and the rogue waves on their background, dynamics in the nonlocal derivative NLS equation does not show any instability or rogue wave phenomena in the class of exact solutions considered in our work.\\
	
	\item Dynamics of the breathers on the traveling periodic waves would naturally appear in the semi-classical limit from initial data with different boundary conditions at infinities. The dispersive hydrodynamics for the nonlocal derivative NLS equation in the focusing case is open for further studies, see \cite{Matsuno} for Whitham's modulation theory in the defocusing case.
\end{enumerate}

\vspace{0.25cm}

{\bf Acknowledgements.} This work was supported in part by the National Natural Science Foundation of China (No. 12371248), the Project “333” of Jiangsu Province, and Southeast University Global Engagement of Excellence Fund (No. 3360682401E).


\begin{thebibliography}{99}
	
\bibitem{Abanov1} A. G. Abanov, E. Betterlheim, and P. B. Wiegmann, ``Integrable hydrodynamics of Calogero--Sutherland model: bidirectional Benjamin--Ono equation", J. Phys. A: Math. Theor. {\bf 42} (2009) 135201 
(24 pages)

\bibitem{Bronsard} Y. Alama Bronsard, X. Chen, M. Dolbeault, ``Spectrally accurate fully discrete schemes for some nonlocal and nonlinear integrable PDEs via explicit formulas", arXiv: 2412.13480 (2024)

\bibitem{B-22} R. Badreddine, ``On the global well-posedness of the Calogero-Sutherland derivative nonlinear Schr\"{o}dinger equation", 
Pure Appl. Anal. {\bf 6} (2024) 379--414
	
\bibitem{B-23} R. Badreddine, ``Traveling waves and finite gap potentials for the Calogero--Sutherland derivative nonlinear Schr\"{o}dinger equation", 
Ann. Inst. H. Poincar\'e C Anal. Non Lineare arXiv: 2307.01592 (2023)

\bibitem{B-24} R. Badreddine, ``Zero dispersive limit of the Calogero--Moser derivative NLS equation", SIAM J. Math. Anal. {\bf 56} (2024) 7228--7249

\bibitem{BLL-22} B. K. Berntson, E. Langmann, and J. Lenells, ``Spin generalizations of the Benjamin--Ono equation", 
Lett. Math. Phys. {\bf 112} (2022) 50 (45 pages)

\bibitem{BF-23} B. K. Berntson and A. Fagerlund, ``A focusing--defocusing intermediate nonlinear Schr\"{o}dinger system", Physica D {\bf 451} 
(2023) 133762 (26 pages)

\bibitem{Bertola} M. Bertola, R. Jenkins, and A. Tovbis, 
``Partial degeneration of finite gap solutions to the Korteweg--de
Vries equation: soliton gas and scattering on elliptic background", 
Nonlinearity {\bf 36} (2023) 3622--3660

\bibitem{Biondini} G. Biondini and D. Mantzavinos, 
``Long-time asymptotics for the focusing nonlinear Schr\"{o}dinger equation with nonzero boundary conditions at infinity and asymptotic stage of modulational instability", Comm. Pure Appl. Math. {\bf 70} (2017) 2300--2365

\bibitem{Biondini2} G. Biondini, S. Li, and D. Mantzavinos, ``Long-time asymptotics for the focusing nonlinear Schr\"{o}dinger equation with nonzero boundary conditions in the presence of a discrete spectrum", 
Comm. Math. Phys. {\bf 382} (2021)  1495--1577

\bibitem{CP-dnls} J. Chen and D. E. Pelinovsky, ``Rogue waves on the background of periodic standing waves in the derivative nonlinear Schr\"{o}dinger equation", Phys. Rev. E {\bf 103} (2021) 062206 (25 pages)

\bibitem{CP-23} J. Chen and D. E. Pelinovsky, ``Periodic waves in the discrete MKDV equation: modulational instability and rogue waves", Physica D {\bf 445} (2023) 133652 (16 pages)

\bibitem{CP-24} J. Chen and D. E. Pelinovsky, ``Rogue waves arising on the standing periodic waves in the Ablowitz-Ladik equation", Stud. Appl. Math. {\bf 152} (2024) 147--173

\bibitem{ChenPel24} J. Chen and D. E. Pelinovsky, ``Bright and dark breathers of the Benjamin-Ono equation on the traveling periodic background", Wave Motion {\bf 126} (2024) 103263 (10 pages)

\bibitem{CPU-dnls} J. Chen, D. E. Pelinovsky, and J. Upsal, ``Modulational instability of periodic standing waves in the derivative NLS equation", J. Nonlin. Sci. {\bf 31} (2021) 58 (32 pages) 

\bibitem{CPW} J. Chen, D. E. Pelinovsky, and R. E. White, ``Periodic standing waves in the focusing nonlinear Schr\"{o}dinger equation: Rogue waves and modulation instability", Physica D {\bf 405} (2020) 132378 (13 pages)

\bibitem{Suret} F. Copie, S. Randoux, and P. Suret, ``The physics 
of the one-dimensional nonlinear Schr\"{o}dinger equation in fiber optics: 
Rogue waves, modulation instability and self-focusing phenomena", 
Rev. Phys. {\bf 5} (2020) 100037 (17 pages)

\bibitem{DK-91} S. Yu. Dobrokhotov and I. M. Krichever, ``Multiphase solutions of the Benjamin--Ono equation and their averaging", Math. Notes {\bf 49} (1991) 583--594

\bibitem{Amin} J. M. Dudley, G. Genty, A. Mussot, A. Chabchoub, and F. Dias, ``Rogue waves and analogies in optics and oceanography", Nat. Rev. Phys. {\bf 1} (2019) 675--689

\bibitem{El}  G. A. El and M. A. Hoefer, ``Dispersive shock waves and modulation theory", Physica D {\bf 133} (2016) 11--65

\bibitem{Hoefer} G. A. El, M. A. Hoefer, and M. Shearer, ``Dispersive and diffusive-dispersive shock waves for nonconvex conservation laws", 
SIAM Rev. {\bf 59} (2017)  3--61

\bibitem{Feng} B. F. Feng, L. Ling, and D. A. Takahashi,
``Multi-breathers and high order rogue
waves for the nonlinear Schr\"{o}dinger equation on the elliptic function
background", Stud. Appl. Math. {\bf 144} (2020) 46--101

\bibitem{Fibich} G. Fibich, {\em The Nonlinear Schrödinger Equation: Singular Solutions and Optical Collapse} (New York: Springer, 2015)


\bibitem{GK-21} P. G\'{e}rard and T. Kappeler, ``On the integrability of the Benjamin--Ono equation on the torus", Comm. Pure Appl. Math. {\bf 74} (2021) 1685--1747

\bibitem{GKT-20} P. G\'{e}rard, T. Kappeler, and P. Topalov, ``On the spectrum of the Lax operator of the Benjamin--Ono equation on the torus", J.  Funct. Anal. {\bf 279} (2020) 108762 (75 pages)

\bibitem{GKT-22} P. G\'{e}rard, T. Kappeler, and P. Topalov, ``On the Benjamin--Ono equation on $\mathbb{T}$ and its periodic and quasiperiodic solutions", J. Spectr. Theor. {\bf 12} (2022) 169--193

\bibitem{GL-22} P. G\'{e}rard and E. Lenzmann, ``The Calogero--Sutherland derivative nonlinear Schr\"{o}dinger equation", Comm. Pure Appl. Math. {\bf 77} (2024) 4008--4062

\bibitem{Girotti} M. Girotti, T. Grava, R. Jenkins, and K. T.-R. McLaughlin,
``Rigorous asymptotics of a KdV soliton gas", 
Comm. Math. Phys. {\bf 384} (2021)  733--784

\bibitem{Grava} M. Girotti, T. Grava, R. Jenkins, K. T.-R. McLaughlin, and 
A. Minakov, ``Soliton versus the gas: Fredholm determinants, analysis, and the rapid oscillations behind the kinetic equation", 
Comm. Pure Appl. Math. {\bf 76} (2023) 3233--3299

\bibitem{HMP} M. Hoefer, A. Mucalica, and D. E. Pelinovsky, ``KdV breathers on cnoidal wave background", J. Phys. A: Math. Theor. {\bf 56} (2023) 185701 (25 pages)

\bibitem{HK} J. Hogan and M. Kowalski, ``Turbulent threshold for continuum Calogero--Moser models", Pure Appl. Anal. {\bf 6} (2024) 941--954.

\bibitem{Kev-Dark-2015} P. G. Kevrekidis, D. J. Frantzeskakis, and R. Carretero-Gonz\'{a}lez, {\em The defocusing nonlinear Schr\"{o}dinger equation: From dark solitons to vortices and vortex rings}, (Society for Industrial and Applied Mathematics, Philadelphia, PA, 2015)

\bibitem{KLV} R. Killip, T. Laurens, and M. Vişan, ``Scaling-critical well-posedness for continuum Calogero--Moser models", arXiv: 2311.12334 (2023)

\bibitem{KKK} K. Kim, T. Kim, and S. Kwon, ``Construction of smooth chiral finite-time blow-up solutions to Calogero--Moser derivative nonlinear Schr\"{o}dinger equation", arXiv:2404.09603 (2024)

\bibitem{KK} T. Kim and S. Kwon, ``Soliton resolution for Calogero–Moser derivative nonlinear Schr\"{o}dinger equation", arXiv:2408.12843 (2024)

\bibitem{Ling} L. Ling and X. Sun, ``Multi-elliptic-dark soliton solutions in the defocusing nonlinear Schr\"{o}dinger equation", Appl. Math. Lett. {\bf 148} (2023) 108866 (9 pages)
	
\bibitem{MatsunoPLA} Y. Matsuno, ``Multiperiodic and multisoliton solutions of a nonlocal nonlinear Schr\"{o}dinger equation for envelope waves'', Phys. Lett. A {\bf 278} (2000) 53--58

\bibitem{MatsunoKaup} Y. Matsuno, `` Linear stability of multiple dark solitary wave solutions of a nonlocal nonlinear Schr\"{o}dinger equation for envelope waves", Phys. Lett. A {\bf 285} (2001) 286--292

\bibitem{Matsuno} Y. Matsuno, ``Asymptotic solutions of the nonlocal nonlinear Schr\"{o}dinger equation in the limit of small dispersion", Phys. Lett. A {\bf 309} (2003) 83--89

\bibitem{MatsunoJCPJ} Y. Matsuno, ``New representations of multiperiodic and multisoliton solutions for a class of nonlocal soliton equations", J. Phys. Soc. Jpn. {\bf 73} (2004) 3285--3293
	
\bibitem{MatsunoSAMP} Y. Matsuno, ``Multiphase solutions and their reductions for a nonlocal nonlinear Schr\"{o}dinger equation with focusing nonlinearity'', Stud. Appl. Math. {\bf 151} (2023) 883--922



\bibitem{Pilod} R. P. de Moura and D. Pilod, ``Local well–posedness for the nonlocal nonlinear Schr\"{o}dinger
equation below the energy space", Adv. Diff. Eqs. {\bf 15} (2010) 925--952

\bibitem{MP-24} A. Mucalica and D. E. Pelinovsky, ``Dark breathers on the snoidal wave background in the defocusing mKdV equation", Lett. Math. Phys. {\bf 114} (2024) 100 (28 pages)

\bibitem{P95} D. E. Pelinovsky, ``Intermediate nonlinear Schro\"{o}dinger equation for internal waves in a fluid of finite depth", Phys. Lett. A {\bf 197} (1995) 401--406

\bibitem{PG95} D. E. Pelinovsky and R. H. J. Grimshaw, ``A spectral transform for the intermediate nonlinear Shr\"{o}dinger equation", J. Math. Phys. {\bf 36} (1995) 4203--4219

\bibitem{PG96} D. E. Pelinovsky and R. H. J. Grimshaw, ``Nonlocal models for envelope waves in a stratified fluid", Stud. Appl. Math. {\bf 97} (1996) 369--391

\bibitem{Shin} H. Shin, ``The dark soliton on a cnoidal wave background", 
J. Phys. A: Math. Gen. {\bf 38} (2005) 3307--3315

\bibitem{Review} A. V. Slunyaev, D. E. Pelinovsky, and E. N. Pelinovsky, ``Rogue waves in the sea: observations, physics and mathematics", Phys. Usp. {\bf 66} (2023) 148--172

\bibitem{Sun-24} R. Sun, ``The intertwined derivative Schr\"{o}dinger system of
Calogero--Moser--Sutherland type", Lett. Math. Phys. {\bf 114} (2024) 74 (32 pages)

\bibitem{Takahashi} D. A. Takahashi, ``Integrable model for density-modulated quantum condensates: Solitons passing through a soliton lattice``, Phys. Rev. E 93 (2016) 062224 (20 pages)

\bibitem{UK} U. Jeong and T. Kim, ``Quantized blow up dynamics for derivative Calogero-Moser nonlinear Schr\"{o}dinger equation", arXiv:2412.12518 (2024)

\bibitem{Voit} M. Voit, ``Freezing limits for Calogero--Moser--Sutherland
particle models", Stud. Appl. Math. {\bf 151} (2023) 1230--1281

\end{thebibliography}
\end{document}